\pdfoutput =1

\documentclass[pra,onecolumn,superscriptaddress]{article}

\usepackage[T1]{fontenc}
\usepackage[utf8]{inputenc}
\usepackage{amsmath,amsfonts,amsthm,amssymb}
\usepackage{mathtools}
\usepackage{subeqnarray}
\usepackage{setspace}
\usepackage{graphics,graphicx,xcolor}
\usepackage{url}
\usepackage{enumitem}
\usepackage{float} 
\usepackage{multirow}
\usepackage{hhline}
\usepackage[margin=1.0in]{geometry}
\usepackage{braket}
\usepackage{dsfont} 
\usepackage{authblk}
\usepackage{multirow}
\usepackage{hhline}
\usepackage[makeroom]{cancel}
\usepackage{ifthen}
\usepackage{comment}
\usepackage{breakcites}

\definecolor{ddgreen}{rgb}{.05,.4,.05}
\definecolor{damethyst}{rgb}{0.4, 0.2, 0.6}
\usepackage[colorlinks, linkcolor=damethyst,citecolor=ddgreen]{hyperref}

\newtheorem{theorem}{Theorem}
\newtheorem{cor}{Corollary}
\newtheorem{definition}{Definition}

\newtheorem{lem}{Lemma}
\newtheorem{prop}{Proposition}
\newtheorem*{prop*}{Proposition}

\newtheorem{conjecture}{Conjecture}

\newtheorem{remark}[theorem]{Remark}

\usepackage[scr=boondoxo,scrscaled=1.05]{mathalfa}

\newcommand{\beq}{\begin{eqnarray}}
\newcommand{\eeq}{\end{eqnarray}}

\DeclareMathOperator{\E}{\mathbf{E}}

\newcommand{\sz}{\sum_{z \in \{0,1\}^r}}
\newcommand{\setr}{\{0,1\}^r}

\newcommand{\tr}{\mathrm{tr}}

\usepackage{tikz}

\definecolor{amethyst}{rgb}{0.6, 0.4, 0.8}

\title{The Power of Two Bases:\\
\Large Robust and copy-optimal certification of\\ nearly all quantum states with few-qubit measurements}
\author[1]{Andrea Coladangelo
}
\author[1]{Jerry Li
}
\author[1]{Joseph Slote
}
\author[2]{Ellen Wu
}

\affil[1]{Paul G. Allen School of Computer Science and Engineering, University of Washington}
\affil[2]{Massachusetts Institute of Technology}
\date{}
\begin{document}

\maketitle
\begin{abstract}
A central task in quantum information science is \textit{state certification}: testing whether an unknown state is $\epsilon_1$-close to a fixed target state, or $\epsilon_2$-far.
Recent work has shown that surprisingly simple measurement protocols---comprising only single-qubit measurements---suffice to certify arbitrary $n$-qubit states \cite{Huang2025,guptaHeODonnell2025}.
However, these certification protocols are not \textit{robust}:
rather than allowing constant $\epsilon_1$,
they can only positively certify states within $\epsilon_1=O(1/n)$ trace distance of the target.
In many experimental settings, the appropriate error tolerance is \textit{constant} as the system size grows, so this lack of robustness renders existing tests inapplicable at scale, no matter how many times the test is repeated.

Here we present robust certification protocols based on few-qubit measurements that apply to all but a $O(2^{-n})$-fraction of pure target states.
Our first protocol achieves \emph{constant} robustness, \textit{i.e.}\ $\epsilon_1=\Theta(1)$, using a single $O(\log n)$-qubit measurement along with single-qubit measurements in the $Z$ or $X$ basis on the other qubits.
As a corollary of its robustness, this protocol also achieves constant (in $n$) copy complexity, which is optimal.
Our second protocol uses exclusively single-qubit measurements and is nearly robust: $\epsilon_1=\Omega(1/\log n)$.
Our tests are based on a new \textit{uncertainty principle for conditional fidelities} which may be of independent interest.
\end{abstract}

\tableofcontents

\newpage 
\section{Introduction}
In \emph{quantum state certification}, the goal is to test whether or not a physical state matches its intended target specification.
This is a fundamental task in quantum information theory with immediate applications to quantum science and engineering, from benchmarking quantum devices and verifying the implementation of quantum algorithms to accurate preparation of lab experiments.
See 
\cite{Badescu2019,ZhuHayashi,KlieschRoth,Huang2025} for more thorough discussions of the broad applicability of the task.

We consider the standard setting of quantum state certification, which also reflects typical experimental setups:  we are given a classical description of an $n$-qubit ``target'' state $\ket{\psi}$ and the ability to produce copies of some physical (potentially mixed) ``lab'' state $\rho$, and our goal is to distinguish between the case where the fidelity between $\ket{\psi}$ and $\rho$ satisfies $\bra{\psi}\rho\ket{\psi} \geq 1 - \epsilon_1$ and the case where $\bra{\psi}\rho\ket{\psi} \leq 1 - \epsilon_2$, for some parameters $\epsilon_1 < \epsilon_2$.

Given the ability to perform arbitrary measurements of the state, the optimal distinguisher is well-understood: we should measure every copy of $\rho$ using the POVM $\{\ket{\psi}\!\!\bra{\psi}, I - \ket{\psi}\!\!\bra{\psi} \}$.
Unfortunately, the ability to perform such a POVM may not only require highly entangled operations, but is in some sense tantamount to preparing $\ket{\psi}$ itself!
In light of this, a recent line of work \cite{pallister2018optimal,Huang2025,theis2024verifying,li2025universal,guptaHeODonnell2025,sater2025efficientcertificationintractablequantum} has asked whether or not one can perform state certification using significantly weaker types of measurements, specifically ones that can be potentially be implemented with high confidence in practice.
In a breakthrough work \cite{Huang2025}, Huang, Preskill, and Soleimanifar first exhibited a state certification test that requires only single-qubit measurements and works for \emph{almost all} target states $\ket{\psi}$.
In subsequent, independent works, Gupta, He, and O'Donnell~\cite{guptaHeODonnell2025} and Li and Zhu~\cite{li2025universal} gave somewhat more complex procedures that allow the user to certify \emph{any} state using a sequence of adaptively chosen single-qubit measurements. In a more recent work, Abdul Sater et al.~\cite{sater2025efficientcertificationintractablequantum} gave a procedure to test a certain class of states using a fully-efficient procedure (including efficient classical processing).
So, does this complete the picture for state certification from simple measurements?

\paragraph{Robust certification.}
Despite the recent exciting progress, existing tests have one fundamental shortcoming: they are not \emph{robust}\footnote{Robust tests are also known as \textit{tolerant tests} in the property testing literature.}.
In particular, they can only certify states that are already $O(1/n)$-close to the target.
This lack of robustness is because they only provide the following guarantees:
\begin{itemize}
\item (Completeness) A state that has $1-\epsilon$ fidelity with the target, passes the test with probability $1-\epsilon$; and,
\item (Soundness) Any state that passes the test with probability $1-\epsilon$ must have $1-O(n \epsilon)$ fidelity with the target.
\end{itemize}
Crucially, the factor of $n$ loss in the soundness means that, \emph{no matter how many copies of the lab state are available}, one cannot hope to distinguish between states that have $0.99$ fidelity and states that have $0.98$ fidelity with the target.
Another way to look at it is the following: unless the lab state is already vanishingly close to the target, \textit{i.e.}\ has $1- O(\frac{1}{n})$ fidelity with the target, the test will not be able to distinguish it from arbitrary states that are potentially very far from the target.
Tests that suffer from this soundness loss are simply not scalable: if an experimentalist wishes to double their system size, the test will only remain usable if they simultaneously double their preparation~accuracy.
What we would ideally like to have instead is a test with only a \emph{constant} soundness loss. We refer to such a test as \emph{robust}\footnote{Our use of the term ``robust'' here is similar to the use in the self-testing literature (which deals with state certification via non-local games), \textit{e.g.}\ in \cite{natarajan2017quantum}. We clarify that this use of the term does not directly consider noise or uncertainty in the measurements involved in the test itself (although better robustness does likely improve the tolerance to measurement noise).}. The latter allows one to certify states of arbitrary size to constant~error. 

Beyond being necessary for scalable certification, robustness also directly affects the copy-complexity of ``amplified'' protocols that are required to distinguish correctly with high probability: a robust ``one-shot'' test only requires a constant number of copies of the lab state in order to amplify the distinguishing success probability to, say, $0.99$.
In light of this discussion, we focus on the following question, which was posed as an open question in \cite{guptaHeODonnell2025}:
\smallskip
\begin{center}
    {\it Is there a robust quantum state certification test that only uses single-qubit measurements?}
\end{center}
\medskip

\subsection{Main results}
In this work, we answer the above question up to logarithmic factors and for all but an exponentially small fraction of states. 

To be a bit more formal, we say that a test is \textit{$c$-robust} if it provides the following guarantee: a state that passes the test with probability $1-\epsilon$, must have fidelity at least $1 - \epsilon/c$ with the target.
Previous results are limited to $c = O(1/n)$.

Our first main result is a state certification algorithm with \emph{constant} robustness for almost all pure target states that uses $n-O(\log n)$ single-qubit measurements (either all in the $Z$ basis, or all in the $X$ basis), and one $O(\log n)$-qubit measurement on the remaining qubits. We clarify at the end of this subsection the type of access/knowledge of the target state that we assume in our tests (and refer to Definition~\ref{def:oracle} for a slightly more formal definition).

\begin{theorem}[Constant robustness with one $O(\log n)$-qubit measurement]
\label{thm:informal-1}
There exists an efficient algorithm that, for all but a $O(2^{-n})$ fraction of pure $n$-qubit target states $\ket{\psi}$, satisfies the following. Given oracle access to a classical description of the target state $\ket{\psi}$ (via the model in Definition~\ref{def:oracle}), and a single copy of an $n$-qubit mixed state $\rho$,
\begin{itemize}
    \item (Completeness) Outputs $\mathsf{accept}$ with probability at least $\bra{\psi}\rho\ket{\psi}$.
    \item (Soundness) Outputs $\mathsf{reject}$ with probability at least
    $\frac{1-\braket{\psi|\rho|\psi}}{2} - o(1)$ (\textit{i.e.}\ approximately the ``infidelity'' up to a factor of $2$).
\end{itemize}
Moreover, the algorithm makes $n-O(\log n)$ single-qubit measurements: either \emph{all} in the $Z$ basis, or \emph{all} in the $X$ basis; and only \emph{one} $O(\log n)$-qubit measurement on the remaining qubits of $\rho$.
\end{theorem}
\noindent There are two surprising aspects of this test (which is described in more detail in Section~\ref{sec:intro-test}): 
\begin{enumerate}
\item \textit{Two bases suffice:} The first $n - O(\log n)$ qubits are measured in one of only \emph{two} bases.
They are either all measured in the $Z$ basis, or all in the $X$ basis, with equal probabilities.
As we explain later, one of the two bases alone is not sufficient to obtain constant robustness.
\item \textit{Minimal randomness:} The test does not actually involve randomizing over the choice of which subset of $O(\log n)$ qubits should be treated differently (in contrast, both the tests from \cite{Huang2025} and \cite{guptaHeODonnell2025} involve randomness in choosing which qubit should be measured differently). Instead, in our test, one can deterministically let this subset be the last $O(\log n)$ qubits.
\end{enumerate}
We emphasize that the constant factor loss in our robustness is only $2$ (this is the factor of $2$ in the soundness condition);
that is, it allows one to distinguish a state with $0.99$ fidelity from a state with $0.98$ fidelity with the target.
Moreover, the $O(\log n)$ number of qubits can be taken to be $(1+\gamma) \log n$, for any constant $\gamma>0$. Finally, we also remark that the latter $O(\log n)$-qubit measurement can be taken to be a \emph{non-adaptive} Clifford measurement sampled uniformly at random, by leveraging classical shadow tomography~\cite{huang2020predicting}, as is done similarly in \cite{Huang2025} (the classical post-processing remains efficient since the measurement is only on $O(\log n)$ qubits).

Via repetition and a Chernoff bound, Theorem~\ref{thm:informal-1} implies the following corollary.
\begin{cor}[Optimal copy-complexity with one $O(\log n)$-qubit measurement per copy]
    There exists an efficient algorithm that, for all but a $O(2^{-n})$ fraction of pure $n$-qubit target states $\ket{\psi}$, satisfies the following. Given oracle access to a classical description of the target state $\ket{\psi}$ (via the model in Definition~\ref{def:oracle}), parameters $\epsilon, \delta > 0$, and $O\left( \epsilon^{-2} \log (1/\delta)\right)$ copies of an arbitrary state $\rho$, 
    \begin{itemize}
        \item (completeness) Outputs \emph{\textsf{accept}} with probability at least $1 - \delta$ if $\bra{\psi} \rho \ket{\psi} \geq 1 - \frac{\epsilon}{2+o(1)}$.
        \item (soundness) Outputs \emph{\textsf{reject}} with probability at least $1 - \delta$, if $\bra{\psi} \rho \ket{\psi} \leq 1 - \epsilon$.
    \end{itemize}
Moreover, for each copy of $\rho$, the algorithm makes $n-O(\log n)$ single-qubit measurements: either \emph{all} in the $Z$ basis, or \emph{all} in the $X$ basis; and only \emph{one} $O(\log n)$-qubit measurement on the remaining qubits of $\rho$.
\end{cor}

Our second main result is a $\Theta(1/\log n)$-robust state certification algorithm for almost all pure target states that uses only single-qubit measurements. This algorithm is obtained by combining our previous algorithm from Theorem~\ref{thm:informal-1}, which includes one $O(\log n)$-qubit measurement, with the test from \cite{guptaHeODonnell2025}, which certifies any state using single-qubit measurements. Since the test from \cite{guptaHeODonnell2025} is not robust, this bootstrapping incurs a robustness loss proportional to the number of qubits to which it is applied, which in our case is $O(\log n)$.
\begin{theorem}[$\Theta(1/\log n)$-robustness with only single-qubit measurements]
\label{thm:informal-2}
There exists an efficient algorithm that, for all but a $O(2^{-n})$ fraction of pure $n$-qubit target states $\ket{\psi}$, satisfies the following. Given oracle access to a classical description of the target state $\ket{\psi}$ (via the model in Definition~\ref{def:oracle}), and a single copy of an $n$-qubit mixed state $\rho$,
\begin{itemize}
    \item (completeness) Outputs $\mathsf{accept}$ with probability at least $\bra{\psi}\rho\ket{\psi}$.
    \item (soundness) Outputs $\mathsf{reject}$ with probability at least $\frac{1-\bra{\psi}\rho\ket{\psi}}{(2+o(1)) \log n}$.
\end{itemize}
Moreover, the algorithm makes only single-qubit measurements.
\end{theorem}
Again, via sequential repetition, and a Chernoff bound, we also have the following.
\begin{cor}[Nearly optimal copy-complexity with only single-qubit measurements]
\label{thm:informal-single-qubit}
    There exists an efficient algorithm that, for all but a $O(2^{-n})$ fraction of pure $n$-qubit target states $\ket{\psi}$, satisfies the following. Given oracle access to a classical description of the target state $\ket{\psi}$ (via the model in Definition~\ref{def:oracle}), parameters $\epsilon, \delta > 0$, and $O\left( \epsilon^{-2} \log^2 (n) \log (1/\delta)\right)$ copies of an arbitrary state $\rho$, 
    \begin{itemize}
        \item (completeness) Outputs \emph{\textsf{accept}} with probability at least $1 - \delta$ if $\bra{\psi} \rho \ket{\psi} \geq 1 - \frac{\epsilon}{(2+o(1))\log n}$.
        \item (soundness) Outputs \emph{\textsf{reject}} with probability at least $1 - \delta$, if $\bra{\psi} \rho \ket{\psi} \leq 1 - \epsilon$.
    \end{itemize}
Moreover, the algorithm makes only single-qubit measurements.
\end{cor}

\noindent

For all of our tests, we assume we have oracle access to the amplitudes of the target state $\ket{\psi}$ in either
the standard or the Hadamard basis (see Definition~\ref{def:oracle} for more details). This type of oracle access is in between the one assumed in \cite{Huang2025} and the one assumed in \cite{guptaHeODonnell2025} in terms of strength. We describe this model a bit more precisely in Section~\ref{sec:access-model}, and comment on its weaknesses.

Up to logarithmic factors, our tests achieve optimal robustness as well as copy-complexity.
In other words, our results demonstrate that one can (at least for most states), \emph{completely match} the performance of the optimal state certification scheme, which may require very entangled measurements, using only single-qubit (or $\log(n)$-qubit) measurements. 

As mentioned, our results only hold for all but an exponentially small $O(2^{-n})$-fraction of target states. This is in contrast to the prior work by~\cite{guptaHeODonnell2025}, whose state certification algorithm works for \emph{all} states. We remark that, in practice, ``nearly all'' states is very different from ``all'' states. This is because many natural states of interest (ground states, Gibbs states, phase states, etc.) are not typical Haar random states. On the other hand, there are also some natural settings where one cares about random (or at least typical) states, \textit{e.g.}\ in demonstrations of quantum advantage based on random circuit sampling~\cite{arute2019quantum}\footnote{While we do not have a formal proof, we conjecture that states generated by random $n^2$-depth brickwork circuits are ``good'' target states for our test with overwhelming probability.}, or based on the preparation of Haar random states~\cite{kretschmer2025demonstrating}; and we are hopeful that our tests can be at least useful as a benchmarking tool to certify the preparation of highly entangled states.

\subsection{Our state certification tests}
\label{sec:intro-test}
Our tests are inspired by the elegant protocol of Huang, Preskill, and Soleimanifar~\cite{Huang2025} (HPS for short).
Their test is conceptually simple: sample a uniformly random $i \in [n]$; measure all qubits of the lab state $\rho$, except for the $i$-th, in the $Z$ basis; let $\ket{\psi_z}$ denote the ``ideal'' post-measurement state of the $i$-th qubit (\textit{i.e.}\ the post-measurement state had $\rho$ been exactly the target), and measure the $i$-th qubit using the POVM $\{ \ket{\psi_z}\bra{\psi_z}, I -  \ket{\psi_z}\bra{\psi_z}\}$. HPS relate the robustness of their protocol to the spectral gap of a certain reversible Markov chain on the hypercube.
They then show that this gap is $\Omega(1/n)$ for all but an exponentially small fraction of $n$-qubit target states (when sampling from the Haar measure).
Unfortunately, this spectral gap really is typically $\Theta(1/n)$---see Appendix~\ref{sec:non-robust-HPS}---so the $O(1/n)$ robustness of their protocol cannot be improved simply via a tighter analysis.
What we find in this work is that, surprisingly, performing an HPS-like test in \textit{both} the standard and Hadamard basis suffices to achieve constant robustness.

Here is a sketch of our main protocol.
Let $C>1$ be a constant arbitrarily close to $1$.
Do one of the following, each with probability 1/2:
\begin{enumerate}
    \item Measure all but the last $C\log n$ qubits of $\rho$ in the \textbf{standard} basis, obtaining some outcome $z\in\{0,1\}^{n-C\log n}$, and post-measurement state $\rho_z$ on the last $C\log n$ qubits.
    Compute the conditional target state $\ket{\psi_z}\propto (\bra{z}\otimes I)\ket{\psi}$.
    Measure $\rho_z$ with the POVM $\{\ket{\psi_z}\!\!\bra{\psi_z}, I -  \ket{\psi_z}\!\!\bra{\psi_z}\}$. Accept or reject according to the measurement outcome.
    \item Measure all but the last $C\log n$ qubits of $\rho$ in the \textbf{Hadamard} basis, obtaining some outcome $z\in\{0,1\}^{n-C\log n}$, and post-measurement state $\hat{\rho}_z$ on the last $C\log n$ qubits.
    Compute the conditional target state $\ket{\hat{\psi}_z}\propto [(\bra{z}H^{\otimes n- C\log n})\otimes I]\ket{\psi}$.
    Measure $\hat{\rho}_z$ with the POVM $\{\ket{\hat{\psi}_z}\!\!\bra{\hat{\psi}_z}, I -  \ket{\hat{\psi}_z}\!\!\bra{\hat{\psi}_z}\}$. Accept or reject according to the measurement outcome.
\end{enumerate}

As we will see, this test obtains \textit{constant} robustness (and accordingly, the optimal constant sample complexity) in exchange for one $C\log n$-qubit measurement.
By replacing this final $C\log n$-qubit measurement with an instance of the one-shot protocol from Gupta, He, and O'Donnell \cite{guptaHeODonnell2025}, we get the result of Theorem~\ref{thm:informal-2}: $\Omega(1/\log n)$ robustness using only single-qubit measurements. 

One may wonder about the test that leaves out only \emph{one} qubit rather than $C\log n$; we discuss this in Section~\ref{sec:future}.

\subsection{Proof ideas}
\label{sec:intro-proof-ideas}
The main conceptual contribution of this work is to establish that, for most states, a local-measurement test in the style of HPS cannot have poor robustness in the standard and Hadamard bases simultaneously.
This intuition can be formalized as an \emph{uncertainty principle for conditional fidelities} which may be of independent~interest.

\paragraph{An uncertainty principle for conditional fidelities.}

Let $\ket{\psi}$ and $\ket{\phi}$ be pure states (which can be thought of as the target and lab state respectively in the context of state certification). Let $\mu_{\phi}^\mathsf{std}$ (resp. $\mu_{\phi}^\mathsf{had}$) denote the distribution over $(n-C\log n)$-bit outcomes from measuring all but the last $C\log n$ qubits of $\ket{\phi}$ in the standard (resp. Hadamard) basis.
Let $\ket{\psi_z}$, $\ket{\phi_z}$, $\ket{\hat{\psi}_z}$, and $\ket{\hat{\phi}_z}$ be $C\log n$-qubit states denoting the post-measurement states resulting from measuring the first $(n-C\log n)$ qubits of  $\ket{\psi}$ (resp. $\ket{\phi}$) in the standard (resp. Hadamard) basis.
Then, we have the following.

\begin{theorem}[Uncertainty principle for conditional fidelities]
\label{thm:uncertainty}
For all but a $O(2^{-n})$-Haar-fraction of $n$-qubit states $\ket{\psi}$, for all $n$-qubit states $\ket{\phi}$,
\begin{align*}
\mathop{\E}_{z\sim \mu_\phi^\mathsf{std}}\big|\big\langle\psi_z|\phi_z\big\rangle\big|^2+\mathop{\E}_{z\sim\mu_\phi^\mathsf{had}}\big|\big\langle\hat{\psi}_z|\hat{\phi}_z\big\rangle \big|^2\;\;&\leq\;\; 1+\,|\!\braket{\psi|\phi}\!|^2+o(1)\,.
\end{align*}
\end{theorem}

Theorem~\ref{thm:uncertainty} is perhaps more recognizable as an uncertainty principle in its equivalent trace distance formulation:
\begin{equation}
\label{eq:td-uncertainty}
\mathop{\E}_{z\sim \mu_\phi^\mathsf{std}} \|\psi_z-\phi_z\|^2_\textnormal{tr}
\;+\;
\mathop{\E}_{z\sim \mu_\phi^\mathsf{had}} \|\hat{\psi}_z-\hat{\phi}_z\|^2_\textnormal{tr}
\;\;\ge\;\;
\|\psi-\phi\|^2_\textnormal{tr}-o(1)\,.\end{equation}
We emphasize that \eqref{eq:td-uncertainty} is a true uncertainty principle in the sense that, while for far-apart $\ket{\psi}$ and $\ket{\phi}$ it is possible for one of the quantities on the left-hand side to vanish, their \textit{sum} may not.
For example, to make the standard basis conditional distance vanish, decompose any target state as $\ket{\psi}=\sum_{z}\alpha_z\ket{z}\ket{\psi_z}$, $\|\alpha\|_2=1$.
Choose any vector $\beta$ with $\|\beta\|_2=1$ and $\langle\beta,\alpha\rangle=0$.
Then $\ket{\phi}:=\sum_z\beta_z\ket{z}\ket{\psi_z}$ is orthogonal to $\ket{\psi}$ yet has $\E_{z\sim\mu_\phi^\mathsf{std}}|\!\braket{\psi_z|\phi_z}\!|^2=1$.
Theorem~\ref{thm:uncertainty} says that for such $\phi$, the expected conditional fidelity in the Hadamard basis must satisfy $\E_{z\in\mu_\phi^\mathsf{had}}|\!\braket{\hat{\psi}_z|\hat{\phi}_z}\!|^2\in o(1)$.

This uncertainty principle underlies both versions of our state certification protocol in the sense that
\[\Pr[\mathsf{accept}] = \frac12 \left(\mathop{\E}_{z\sim \mu_\phi^\mathsf{std}}\big|\big\langle\psi_z|\phi_z\big\rangle\big|^2+\mathop{\E}_{z\sim\mu_\phi^\mathsf{had}}\big|\big\langle\hat{\psi}_z|\hat{\phi}_z\big\rangle \big|^2\right).\]
Consequently, Theorem~\ref{thm:uncertainty} implies the following upper bound:
\[\Pr[\mathsf{accept}]\leq \tfrac12+\tfrac12|\!\braket{\psi|\phi}\!|^2 + o(1)\, ,\]
which is equivalent to the ``soundness'' condition in our main theorem.

\medskip

The uncertainty principle of Theorem~\ref{thm:uncertainty} is a rephrasing of our Theorem~\ref{thm:2}.
While it is conceptually intuitive and simple to state, its proof turns out to be quite technical.
The first key step is to find a way to write the Hadamard basis expression $\mathop{\E}_{z\sim\mu_\phi^\mathsf{had}}\big|\big\langle\hat{\psi}_z|\hat{\phi}_z\big\rangle \big|^2$ in terms of both the ``standard basis'' post-measurement states $\ket{\psi_z}$, $\ket{\phi_z}$, and the ``global'' inner product $\braket{\psi|\phi}$.
To control the resulting expression we rely on sophisticated machinery from random matrix theory.

\subsection{Discussion and future work}
\label{sec:future}
This work gives an $\Omega(1/\log n )$-robust algorithm for verifying Haar-typical states using single-qubit measurements only, as well as a $\Theta(1)$-robust algorithm using single-qubit measurements along with one $O(\log n)$-qubit measurement.
There are several exciting next steps for quantum state certification using simple measurements, as we outline next.

\paragraph{Constant robustness with single-qubit measurements.}
We conjecture that the version of our protocol that leaves out only a \emph{single} qubit (chosen uniformly at random), rather than $O(\log n)$ qubits, has constant robustness for Haar-typical target states.

\begin{conjecture}
    Consider a modified version of our main protocol from Section~\ref{sec:intro-test} that measures all but \textit{one} qubit---chosen uniformly at random---and then applies the ideal post-selected measurement on the final qubit.
    This test has constant robustness.
\end{conjecture}

In preliminary simulations, we have found strong evidence for the validity of this conjecture.
Note that if this conjecture were true, it would give a fully non-adaptive test, obtained by replacing the ideal one-qubit measurement by a random Pauli measurement (leveraging classical shadow tomography).
This would also show---at least for all but a negligible fraction of states---that the statistically optimal measurement is matched in performance (robustness, copy complexity) by a protocol that only uses single-qubit measurements. We also remark that our numerical simulations suggest that randomizing over the choice of qubit to ``leave out'' is essential for robustness.
This is in contrast to our test that uses one $O(\log n)$-qubit measurement, where, perhaps surprisingly, randomization over the choice of subset of $O(\log n)$ qubits is not needed.

\paragraph{Beyond typicality.}
The focus of this work was to obtain a high degree of robustness.
Although the tests presented in this work apply to all but a negligible fraction of hypothesis states, it is important to (\textit{i}) understand whether explicit families of states important in practice (phase states, ground states of certain Hamiltonians, etc.) are covered, and (\textit{ii}) whether our protocol can be extended or modified to verify \textit{all} states. 

There are some simple classes of target states that are not covered by our test, but are covered by other basic tests.
For example, if the target state had the form $\ket{\psi}=\ket{+i}\otimes\ket{\psi_\text{rest}}$, then a lab state of the form $\ket{\phi}=\ket{-i}\otimes \ket{\psi_\text{rest}}$
would be accepted with probability 1 although $\ket{\psi}$ and $\ket{\phi}$ are orthogonal (here $\ket{+i}$ and $\ket{-i}$ are the $\pm1$ eigenvectors of Pauli $Y$).
Of course, this pair of states can be trivially distinguished by the appropriate measurement on the first qubit.

Since our test can be made non-adaptive via shadow tomography, the adaptivity lower bound of \cite{guptaHeODonnell2025} suggests a new idea will be needed to obtain robustness for all states.
Robust certification of \textit{all} quantum states by single-qubit measurements remains a tantalizing open problem.

One may wonder if running one-shot GHO
    \cite{guptaHeODonnell2025} on both $\ket{\psi}$ and $H^{\otimes n}\ket{\psi}$  would suffice.
    Unfortunately, this will also not work, as the following proposition shows.
    
\begin{prop}[No universal robustness for ``conditional comparison tests'']
\label{prop:no-univ-robustness}
    Let $r < n$. Consider any state certification algorithm with the following structure.
    \begin{enumerate}[label=\textit{\roman*}.]
        \item[(i)] Measure all but $r$ qubits with single-qubit measurements $M_1,\ldots, M_{n-r}$ (where any measurement order and adaptivity is allowed).
        \item[(ii)] Measure the remaining $r$ qubits with a measurement $M^*$ depending arbitrarily on the outcomes of $M_1,\ldots, M_{n-r}$.
        \item[(iii)] Output \textsf{accept} or \textsf{reject} depending exclusively on the outcome of $M^*$.
    \end{enumerate}
    Then, provided the algorithm accepts the target state $\ket{\psi}$ itself with certainty, there exist target states $\ket{\psi}$ for which the algorithm has robustness $O(r/n)$.
\end{prop}
The proof is simple and included in Appendix~\ref{sec:no-univ-robustness}. This suggests that a test with constant robustness for \textit{all} target states using local measurements would require a decision procedure that takes into account the statistics of all the measurements, not just the final measurement $M^*$.
In particular, it would need to deviate substantially from existing paradigms for state certification with one-qubit measurements.

\paragraph{Beyond pure states.} Another interesting direction is to ask whether or not one can extend the algorithms for state certification with single or few-qubit measurements beyond the setting where the target state is a pure state.
Note that the analysis of our algorithm, as well as the prior work of~\cite{Huang2025,guptaHeODonnell2025} crucially require that the target state is pure.
On the other hand, it is known (see \textit{e.g.}~\cite{o2015quantum}) that exponentially many copies are needed for state certification of arbitrary mixed states, even with arbitrarily powerful measurements, and thus some structural assumptions are required for us to hope for efficient algorithms.
A natural open question thus is whether we can still obtain efficient and robust algorithms for state certification with single-qubit measurements for sufficiently simple classes of mixed target states, such as random low-rank states, or Gibbs states of simple Hamiltonians.

\subsection{Access model}
\label{sec:access-model}
As in earlier protocols \cite{Huang2025,guptaHeODonnell2025}, we assume classical access to an oracle-based description of the target state $\psi$. The access model that we work with is, in terms of strength,  between that of HPS and GHO.
In particular, we require knowledge of amplitudes of the target state in \emph{two} bases: standard and Hadamard.
For comparison, HPS only requires amplitudes in the standard basis, whereas GHO require amplitudes in all product bases (but is capable of certifying \emph{all} target states, rather than ``almost all'').

We acknowledge that our access model creates a bit of a tension: Haar-random states are highly entangled and exponentially complex, whereas states for which this kind of amplitude access can be simulated efficiently are generally more structured and tractable. For example, tensor network states with polynomial bond dimension allow efficient amplitude queries but are low-entangled, and neural quantum states typically support efficient queries in only one basis. This means that our test will generally require inefficient classical post-processing. 
We note that, while GHO require an even stronger access model, their certification works for all states (which thus includes classes of states for which the access model is efficiently reproducible). That said, we find the existence of a robust test that uses few-qubit measurements surprising even from a purely information-theoretic standpoint.
\begin{definition}
\label{def:oracle}
We consider an oracle model which accepts queries in the format $(z,\mathsf{b})$ for $z\in\{0,1\}^n$ and $\mathsf{b}\in\{\mathsf{std},\mathsf{had}\}$\,, and returns the complex number
\[\mathcal{O}_\psi(z,\mathsf b)=\begin{cases}
\braket{z|\psi} & \text{if } \mathsf{b}=\mathsf{std}\\[0.5em]
\braket{z|H^{\otimes n}|\psi} & \text{if } \mathsf{b}=\mathsf{had}\,.
\end{cases}\]
\end{definition}

This access model allows us to explicitly and efficiently learn the leftover state on $O(\log n)$ qubits conditioned on the outcome of any measurement on $\ket{\psi}$ used in our protocols.
\begin{lem}
    Let $C>1$ be a constant and consider the standard- and Hadamard-basis decompositions of an $n$-qubit state $\ket{\psi}$:
    \[\ket{\psi}=\sum_{z\in\{0,1\}^{n-C\log n}}\ket{z}\ket{\psi_z}=H^{\otimes n}\Big(\sum_{z\in\{0,1\}^{n-C\log n}}\ket{z}\ket{\hat{\psi}_z}\Big).\]
    Then, using access to $\mathcal{O}_\psi$ as in Definition~\ref{def:oracle}, it is possible to obtain a complete description of any of the $\ket{\psi_z}$'s or the $\ket{\hat{\psi}_z}$'s in polynomial time.
\end{lem}
\begin{proof}
    We can learn the conditional standard basis state $\ket{\psi_z}$, for some $z\in\{0,1\}^{n-C\log n}$, by querying $\mathcal O_{\psi}(z||w,\mathsf{std})$ for all $w\in\{0,1\}^{C\log n}$, for a total query complexity (and runtime) of $n^C$.
\end{proof}
As a consequence of this lemma, we are also able to simulate the $O(\log n)$-qubit version of the oracle needed for the one-shot GHO subroutine needed for Theorem~\ref{thm:informal-2}.

\subsection{Additional related work}
Beyond the aforementioned line of work on state certification with single-qubit measurements, there is by now a vast literature on quantum state certification in various other settings. Here we will focus only on the most relevant works.
Quantum state certification was first considered in the quantum information theory community in the asymptotic regime, where it was referred to as \emph{quantum state discrimination}; see \textit{e.g.}~\cite{audenaert2008asymptotic} for a survey of these results.
More recently, there have been a number of works studying the copy complexity of quantum state certification of mixed states, see \textit{e.g.}~\cite{o2015quantum,Badescu2019,haah2016sample,flammia2011direct,da2011practical,aolita2015reliable,chen2022toward,o2025instance}.
However, the methods employed either required heavily entangled and complex measurements, or a number of copies that scales exponentially with the system size, or both, and it is known that certification of general mixed states requires exponentially many copies, even when arbitrarily complex measurements are allowed~\cite{childs2007weak,o2015quantum}.
Efficient state certification algorithms that use simpler measurements have also been developed for several classes of special states, including stabilizer states~\cite{flammia2011direct,huang2020predicting,huang2024learning}, bosonic and fermionic Gaussian states~\cite{aolita2015reliable,gluza2018fidelity}, and states generated by specific types of circuits~\cite{takeuchi2018verification,huang2024learning}.
Yet other methods achieve good fidelity estimation, but only under specific noise models~\cite{arute2019quantum,choi2023preparing,cotler2023emergent}.
We also mention that state certification itself fits into a larger literature on quantum property testing and learning, see \textit{e.g.}~\cite{montanaro2013survey,anshu2024survey} for a more in-depth discussion. 

Finally, the problem of state certification also has long tradition in the setting of Bell's inequalities, where it goes by the name of \emph{self-testing}~\cite{mayers2004self}. We refer to \cite{vsupic2020self} for a survey.

\section*{Acknowledgments}
The authors are grateful to Chinmay Nirkhe and Yihui Quek for various discussions during this work.
The authors also thank the anonymous STOC reviewers for helpful comments that improved the manuscript. A.C.\ is thankful for support from the Google Research Scholar program. E.W.\ was supported by the Quantum@UW REU program, funded by NSF award CCF-2446908. 

\section{Robust certification with one logarithmic-qubit measurement}
\subsection{The test}
\label{sec:algo-log}
Let $\ket{\psi}$ be the target state, and $\rho$ be the lab state. Let $\gamma>0$ be an arbitrary constant. Let $r = n - (1+\gamma)\log n$. 
\begin{itemize}
\item With probability $\frac12$, measure the first $r$ qubits of $\rho$ in the \emph{standard} basis. Let $z \in \{0,1\}^r$ be the outcome, and let $\rho_z$ be the post-measurement state on the last $(1+\gamma)\log n$ qubits. Compute the ``ideal'' post-measurement state $\ket{\psi_z}\propto (\bra{z}\otimes I)\ket{\psi}$. Measure the last $(1+\gamma)\log n$ qubits in the basis $\{ \ket{\psi_z}\bra{\psi_z}, I - \ket{\psi_z}\bra{\psi_z}\}$. Accept if the first outcome is obtained, and reject otherwise.
\item With probability $\frac12$, measure the first $r$ qubits of $\rho$ in the \emph{Hadamard} basis. Let $z \in \{0,1\}^r$ be the outcome, and let $\hat{\rho}_z$ be the post-measurement state on the last $(1+\gamma)\log n$ qubits. Compute the ``ideal'' post-measurement state $\ket{\hat{\psi}_z}\propto [(\bra{z}H^{\otimes n- C\log n})\otimes I]\ket{\psi}$. Measure the last $(1+\gamma)\log n$ qubits in the basis $\{ \ket{\hat{\psi}_z}\bra{\hat{\psi}_z}, I - \ket{\hat{\psi}_z}\bra{\hat{\psi}_z}\}$. Accept if the first outcome is obtained, and reject otherwise.
\end{itemize}

\begin{remark}
    We note that this test can be made both (\textit{i}) simpler to implement and (\textit{ii}) fully non-adaptive by replacing the optimal $O(\log n)$-qubit measurement with a random Clifford measurement and applying the classical shadows framework \cite{huang2020predicting} to estimate the overlap with the target conditional state after the fact.
\end{remark}

\subsection{Proof of robustness}
Completeness of the test follows from the fact, for any lab state $\ket{\phi}$, its probability of acceptance is $\textnormal{Tr}\left[A \ket{\phi}\!\!\bra{\phi}\right]$, where $A$ is the PSD operator $$A = \frac12 \sum_{z} \ket{z}\!\!\bra{z}\otimes  \ket{\psi_z}\!\!\bra{\psi_z} + \frac12  \sum_{z} (H^{\otimes r}\ket{z}\!\!\bra{z} H^{\otimes r})\otimes\ket{\hat{\psi}_z}\!\!\bra{\hat{\psi}_z} \,.$$
Expanding $\ket{\phi}$ in an eigenbasis of $A$, we see that $\textnormal{Tr}\left[A \ket{\phi}\bra{\phi}\right]$ is a convex combination of the contributions from each eigenvector (since the eigenvectors are orthogonal to each other). Since $\ket{\psi}$ itself is an eigenvector with eigenvalue $+1$, we have $\textnormal{Tr}\left[A \ket{\phi}\!\!\bra{\phi}\right] \geq |\braket{\psi| \phi}|^2$. This extends to mixed lab states $\rho$, to give $\textnormal{Tr}[A \rho] \geq \bra{\psi} \rho \ket{\psi}$.  

\medskip
So, from here on, we focus on proving soundness of our test. For convenience of notation, we will take $\gamma = 1$ in this proof, so $r = n-2\log n$. However, the proof works for any $\gamma >0$. We will show the following.

\begin{theorem}
\label{thm:1}
With probability $1 -O(2^{-n})$ over sampling a Haar random $n$-qubit target pure state $\ket{\psi}$, the test from Section~\ref{sec:algo-log} has the following guarantee: for any $\epsilon \geq 0$, if a (potentially mixed) state $\rho$ passes the test with probability $1-\epsilon$, then $\braket{\psi|\rho| \psi} \geq 1-\big( 2+o(1)\big) \epsilon$.
\end{theorem}

To show the above, it suffices to show the following.
\begin{theorem}
\label{thm:2}
With probability $1 -O(2^{-n})$ over sampling a Haar random $n$-qubit target pure state $\ket{\psi}$, the following holds: for all states $\ket{\phi}$ \emph{orthogonal} to $\ket{\psi}$, the probability that $\ket{\phi}$ passes the test from Section~\ref{sec:algo-log} is at most $\frac12 +o(1)$.
\end{theorem}

\begin{proof}[Proof of Theorem~\ref{thm:1} from Theorem~\ref{thm:2}.]
    Let $A$, $0\preceq A \preceq I$, be the POVM representing the protocol. Then, $\braket{\psi|A|\psi}=1$, and, by Theorem~\ref{thm:2}, with probability $1 -O(2^{-n})$ over the target state $\ket{\psi}$, the following holds for all $\ket{\phi}$ orthogonal to $\ket{\psi}$: $\braket{\phi|A|\phi}\leq \frac12 + c$, where $c =o(1)$.
    Then, $A\preceq \ket{\psi}\!\!\bra{\psi}+(\frac12+c)P$ for $P=I-\ket{\psi}\!\!\bra{\psi}$.
    That is,
    \[A\preceq \left(\frac12+c\right)I + \left(\frac12 -c\right)\ket{\psi}\!\!\bra{\psi}\,.\]
    Now, from our assumption on $\rho$, we have
    \[1-\epsilon\leq \tr[\rho A] \leq \frac12 + c + \left(\frac12-c\right)\braket{\psi|\rho|\psi} \,.\]
    Rearranging gives $\braket{\psi|\rho|\psi}\geq \frac{\frac12 -\epsilon -c }{\frac12 - c} = 1- \frac{2\epsilon}{1-2c} = 1-\big( 2+o(1)\big) \epsilon$.
\end{proof}

The rest of the section is dedicated to proving Theorem~\ref{thm:2} (and hence Theorem~\ref{thm:1}). 

For convenience and readability of notation, we denote the target state as $\ket{f}$, where $f: \{0,1\}^n \rightarrow \mathbb{C}$ such that $\ket{f} = \sum_{x \in \{0,1\}^n} f(x) \ket{x}$. So $\sum_{x} |f(x)|^2 = 1$. We denote the lab state as $\ket{g}$ where $g: \{0,1\}^n \rightarrow \mathbb{C}$ and $\ket{g} = \sum_{x \in \{0,1\}^n} g(x) \ket{x}$. So $\sum_{x} |g(x)|^2 = 1$. We suppose $\langle f|g \rangle = 0$.

We can write the probability that $\ket{g}$ passes the test as $\frac12 \Pr[\textnormal{accept}|\textsf{std}] + \frac12 \Pr[\textnormal{accept}|\textsf{had}]$, where $\Pr[\textnormal{accept}|\textsf{std}]$ is the probability $\ket{g}$ passes the ``standard basis test'', and $\Pr[\textnormal{accept}|\textsf{had}]$ is the probability that $\ket{g}$ passes the ``Hadamard basis test''. 

Now, we write 
$$\ket{f} = \sum_{z \in \{0,1\}^r} \ket{z}\ket{f_z}\,,$$ where the $\ket{f_z}$ are $(n-r)$-qubit states satisfying $\sum_{z \in \{0,1\}^r} \| \ket{f_z}\|^2 = 1$. Further, we denote $\ket{\hat{f}} = H^{\otimes n} \ket{f}$, and write $$\ket{\hat{f}} = \sum_{z \in \{0,1\}^r} \ket{z}\ket{\hat{f}_z}\,,$$ 
where $\sum_{z \in \{0,1\}^r} \| \ket{\hat{f}_z}\|^2 = 1$. We use analogous notation for $\ket{g}$.
\begin{remark}
The attentive reader might notice that our use of the notation $\hat{f}_z$ for the conditional state on the last $n-r$ qubits in the ``Hadamard'' basis differs from the use of $\hat{\psi}_z$ in the introduction (Sections~\ref{sec:intro-test} and \ref{sec:intro-proof-ideas}). The two uses differ only by an application of Hadamards on the $n-r$ qubits. We find the use of $\hat{f}_z$ described here to be the more natural one: it is the leftover state of the last $n-r$ qubits conditioned on obtaining outcome $z$ when measuring the first $r$ qubits of $\ket{\hat{f}} = H^{\otimes n} \ket{f}$. 
\end{remark}
With this notation in place, observe that we have
\begin{align}
\Pr[\textnormal{accept}|\textsf{std}] &= \sum_{z \in \{0,1\}^r} \frac{|\langle f_z, g_z \rangle |^2}{\| \ket{f_z}\|^2} \,, \label{eq:score_std}\\
\Pr[\textnormal{accept}|\textsf{had}] &= \sum_{z \in \{0,1\}^r} \frac{|\langle \hat{f}_z, \hat{g}_z \rangle |^2}{\| \ket{\hat{f}_z}\|^2} \,. \label{eq:score_had}
\end{align}
Now, it will be convenient, for each $z \in \{0,1\}^r$, to write 
\begin{equation}
\ket{g_z} = c_z \frac{\ket{f_z}}{\|f_z\|} + c^{\perp}_z \ket{f^{\perp}_z} \,, \label{eq:3}
\end{equation} for some complex numbers $c_z, c^{\perp}_z$, and a (normalized) $(n-r)$-qubit state $\ket{f^{\perp}_z}$, \textit{i.e.}\ we are writing $\ket{g_z}$ as a sum of a component that is parallel to the ``ideal'' vector, and a component orthogonal to it. Then, substituting this into \eqref{eq:score_std} gives
\begin{equation}
\Pr[\textnormal{accept}|\textsf{std}] = \sum_{z \in \{0,1\}^r} |c_z|^2 \,. \label{eq:score_std_2}
\end{equation}

Recall that our goal is to show that the acceptance probability $\frac12 \Pr[\textnormal{accept}|\textsf{std}] + \frac12 \Pr[\textnormal{accept}|\textsf{had}]$ is bounded away from 1 by a constant. In other words, we want to show that there is a constant $\delta>0$ such that at least one of $\Pr[\textnormal{accept}|\textsf{std}]$ or $\Pr[\textnormal{accept}|\textsf{had}]$ is at most $1-\delta$. We will thus have to derive an ``uncertainty principle''-like tradeoff between $\Pr[\textnormal{accept}|\textsf{std}]$ and $\Pr[\textnormal{accept}|\textsf{had}]$. One crucial point is that, if we allow $\ket{g} = \ket{f}$, then both expressions can of course be made equal to $1$. Instead, we are assuming $\braket{f|g} = 0$, and so our tradeoff (and our argument) will have to crucially rely on this. We will derive the following expression for $\Pr[\textnormal{accept}|\textsf{had}]$ in terms of the $c_z, c^{\perp}_z$, and the $\ket{f_z}$, $\ket{f^{\perp}_z}$ (and their Fourier transforms). Importantly, this expression isolates precisely the term $\braket{f|g}$.
\begin{lem}
\label{lem:tradeoff}
$$\Pr[\textnormal{accept}|\mathsf{had}] = \frac{1}{2^{2r}} \sum_{z \in \{0,1\}^r} \frac{1}{\|\ket{\hat{f}_z}\|^2} \left| \braket{f|g} + \sum_{w \neq w' \in \{0,1\}^r} (-1)^{z \cdot (w \oplus w')} \Big(c_w \frac{\braket{f_{w'}| f_w}}{\|\ket{f_w}\|} + c^{\perp}_w \braket{f_{w'}| f^{\perp}_w}\Big)\right|^2 \,.$$
\end{lem}

We start by deriving an expression for the $\ket{\hat{f}_z}$ and $\ket{\hat{g}_z}$.
\begin{lem}
\label{lem:hat}
Let $r \in [n]$. Let $\ket{h}$ be an $n$-qubit state. Write
$$ \ket{h} = \sum_{z \in \{0,1\}^r} \ket{z}\ket{h_z}$$ for some $\ket{h_z}$ such that $\sum_{z \in \{0,1\}^r} \|\ket{h_z}\|^2 = 1$. Let $\ket{\hat{h}} = H^{\otimes n} \ket{h}$, and write
$$ \ket{\hat{h}} = \sum_{z \in \{0,1\}^r} \ket{z}\ket{\hat{h}_z}$$ for some $\ket{\hat{h}_z}$ such that $\sum_{z \in \{0,1\}^r} \|\ket{\hat{h}_z}\|^2 = 1$. Then,
for all $z \in \{0,1\}^r$,
$$ \ket{\hat{h}_z} = \frac{1}{\sqrt{2^r}}\sum_{w \in \{0,1\}^r} (-1)^{z\cdot w} \ket{\widehat{h_w}}\,,$$
where $\ket{\widehat{h_w}}= H^{\otimes n-r} \ket{h_r}$.
\end{lem}
\begin{proof}
We have 
\begin{align*}
    \ket{\hat{h}} = H^{\otimes n} \sz \ket{z}\ket{h_z} &= \sz (H^{\otimes r}\ket{z}) (H^{\otimes n-r} \ket{h_z}) \\ &= \sz (H^{\otimes r}\ket{z}) \ket{\widehat{h_z}} \\
    &= \frac{1}{\sqrt{2^r}}\sum_{w,z \in \{0,1\}^r} (-1)^{z\cdot w} \ket{w} \ket{\widehat{h_z}} \\
    &=\frac{1}{\sqrt{2^r}} \sz \ket{z} \Big(\sum_{w \in \{0,1\}^r} (-1)^{z\cdot w}\ket{\widehat{h_w}}\Big) \,.
\end{align*}
Thus, we have $\ket{\hat{h}_z} = \frac{1}{\sqrt{2^r}}\sum_{w \in \{0,1\}^r} (-1)^{z\cdot w}\ket{\widehat{h_w}}$.
\end{proof}
We can now prove Lemma~\ref{lem:tradeoff}.

\begin{proof}[Proof of Lemma \ref{lem:tradeoff}]
Recall from \eqref{eq:score_had} that 
$$\Pr[\textnormal{accept}|\textsf{had}] = \sum_{z \in \{0,1\}^r} \frac{|\langle \hat{f}_z, \hat{g}_z \rangle |^2}{\| \ket{\hat{f}_z}\|^2} \,.$$
From Lemma~\ref{lem:hat}, we have
\begin{align}
\ket{\hat{f}_z} &= \frac{1}{\sqrt{2^r}}\sum_{w \in \{0,1\}^r} (-1)^{z\cdot w} \ket{\widehat{f_w}} \nonumber\\
\ket{\hat{g}_z} &= \frac{1}{\sqrt{2^r}}\sum_{w \in \{0,1\}^r} (-1)^{z\cdot w} \ket{\widehat{g_w}} \,. \nonumber
\end{align}
Moreover, by definition of $\ket{g_w}$ (from Equation~\ref{eq:3}) and linearity of $H^{\otimes n-r}$, we have $\ket{\widehat{g_w}} = c_w \frac{\ket{\widehat{f_w}}}{\|\ket{f_w} \|} + c^{\perp}_w \ket{\widehat{f^{\perp}_w}}$.
Thus, 
\begin{align}
\Pr[\textnormal{accept}|\textsf{had}] &= \frac{1}{2^{2r}} \sz \frac{1}{\| \ket{\hat{f}_z}\|^2} \left| \sum_{w,w' \in \setr} (-1)^{z \cdot (w\oplus w')} \langle \widehat{f_{w'}}, \widehat{g_w} \rangle \right|^2 \nonumber \\
&= \frac{1}{2^{2r}} \sz \frac{1}{\| \ket{\hat{f}_z}\|^2} \left| \sum_{w,w' \in \setr} (-1)^{z \cdot (w\oplus w')} \left(  c_w \frac{\braket{\widehat{f_{w'}}| \widehat{f_w}}}{\| \ket{f_w}\|} + c_w^{\perp} \braket{\widehat{f_{w'}}| \widehat{f^{\perp}_w}} \right) \right|^2 \nonumber \\
&=\frac{1}{2^{2r}} \sz \frac{1}{\| \ket{\hat{f}_z}\|^2} \left| \sum_{w,w' \in \setr} (-1)^{z \cdot (w\oplus w')} \left(  c_w \frac{\braket{f_{w'}| f_w}}{\| \ket{f_w}\|} + c_w^{\perp} \braket{f_{w'}| f^{\perp}_w} \right) \right|^2 \nonumber  \\
&=  \frac{1}{2^{2r}} \sz \frac{1}{\| \ket{\hat{f}_z}\|^2} \left| \sum_{w \in \setr} c_w \| \ket{f_w}\| + \sum_{w\neq w' \in \setr} (-1)^{z \cdot (w \oplus w')}  \left(  c_w \frac{\braket{f_{w'}| f_w}}{\| \ket{f_w}\|} + c_w^{\perp} \braket{f_{w'}| f^{\perp}_w} \right)  \right|^2 \label{eq:5}
\end{align}
where the third equality follows by the fact that $H^{\otimes n-r}$ preserves inner products. Finally, observe that we have precisely
\begin{equation}
\label{eq:122}
\braket{f|g} = \sum_{w \in \setr} \braket{f_w | g_w} = \bra{f_w} (c_w \frac{\ket{f_w}}{\| \ket{f_w}\|}+ c^{\perp}_w \ket{f^{\perp}_w} ) = \sum_{w \in \setr} c_w \| \ket{f_w}\| \,. 
\end{equation}
Substituting the latter into \eqref{eq:5} gives the desired equality.
\end{proof}
\begin{cor}
When $\braket{f|g}=0$, 
$$\Pr[\textnormal{accept}|\mathsf{had}] = \frac{1}{2^{2r}} \sum_{z \in \{0,1\}^r} \frac{1}{\|\ket{\hat{f}_z}\|^2} \left| \sum_{w \neq w' \in \{0,1\}^r} (-1)^{z \cdot (w \oplus w')} \Big(c_w \frac{\braket{f_{w'}| f_w}}{\|\ket{f_w}\|} + c^{\perp}_w \braket{f_{w'}| f^{\perp}_w}\Big)\right|^2 \,.$$
\end{cor}
We will now expand the squared magnitude on the RHS above, and spend the rest of the proof obtaining a bound on the resulting terms. We have
\begin{align}
\Pr[\textnormal{accept}|\mathsf{had}] = \frac{1}{2^{2r}} \sz \frac{1}{\|\ket{\hat{f}_z}\|^2} \cdot \sum_{\substack{w\neq w' \\ v\neq v'}} (-1)^{z \cdot (w\oplus w' \oplus v \oplus v')} 
 \bigg[ &\overline{c_w} c_v \overline{A_{w,w'}}\cdot A_{v,v'} \label{eq:6}\\
+&2 \textnormal{Re}(\overline{c_w}c_v^{\perp} \overline{A_{w,w'}}\cdot B_{v,v'}) \label{eq:7}\\
+& \overline{c^{\perp}_w}c_v^{\perp} \overline{B_{w,w'}}\cdot B_{v,v'} \bigg] \,, \label{eq:8}
\end{align}
where we defined $A_{w,w'} = \frac{\braket{f_{w'}| f_w}}{\|\ket{f_w}\|}$ and $B_{w,w'} = \braket{f_{w'}| f^{\perp}_w}$.
We now proceed to bound each of the terms \eqref{eq:6}, \eqref{eq:7}, and \eqref{eq:8}. We start with \eqref{eq:8}. 
\begin{align}
\eqref{eq:8} &= \frac{1}{2^{2r}} \sz \frac{1}{\|\ket{\hat{f}_z}\|^2} \cdot \sum_{\substack{w\neq w' \\ v\neq v'}} (-1)^{z \cdot (w\oplus w' \oplus v \oplus v')} 
\overline{c^{\perp}_w} c^{\perp}_v \overline{B_{w,w'}} B_{v,v'} \nonumber \\
&= \frac{1}{2^{2r}} \sz \frac{1}{\|\ket{\hat{f}_z}\|^2} \cdot \sum_{\substack{w, w' \\ v, v'}} (-1)^{z \cdot (w\oplus w' \oplus v \oplus v')} 
\overline{c^{\perp}_w} c^{\perp}_v \overline{B_{w,w'}} B_{v,v'} \,, \nonumber
\end{align}
where in the last line the sum is over all $w,w',v,v'$ (rather than just $w\neq w'$ and $v\neq v'$), and the equality follows from the fact that $B_{w,w} = \braket{f_w | f^{\perp}_w} = 0 $. Continuing from above, we have
\begin{align}
\eqref{eq:8} & = \frac{1}{2^{2r}} \sz \frac{1}{\|\ket{\hat{f}_z}\|^2} \cdot \left| \sum_{w,w'} (-1)^{z \cdot (w\oplus w')} c^{\perp}_w B_{w,w'} \right|^2 \label{eq:99}\\
&= \frac{1}{2^{2r}} \sz \frac{1}{\|\ket{\hat{f}_z}\|^2} \cdot \bigg| \sum_{w,w'} (-1)^{z \cdot (w\oplus w')} c^{\perp}_w  \braket{f_{w'}| f^{\perp}_w} \bigg|^2 \nonumber\\
&= \frac{1}{2^{2r}} \sz \frac{1}{\|\ket{\hat{f}_z}\|^2} \bigg| \Big( \sum_{w'} (-1)^{z \cdot w'} \bra{f_{w'}} \,\,\sum_w (-1)^{z\cdot w} c^{\perp}_w \ket{f^{\perp}_w} \Big) \bigg|^2 \nonumber \\
&\leq \frac{1}{2^{2r}} \sz \frac{1}{\|\ket{\hat{f}_z}\|^2} \Big\| \sum_{w'} (-1)^{z\cdot w'} \ket{f_{w'}} \Big\|^2 \cdot \Big\| \sum_w (-1)^{z\cdot w} c^{\perp}_w \ket{f^{\perp}_w} \Big\|^2 \nonumber\\
&= \frac{1}{2^{2r}} \sz \frac{1}{\|\ket{\hat{f}_z}\|^2} \|\ket{F_z} \|^2 \cdot \| \ket{G_z} \|^2 \,, \label{eq:13}
\end{align}
where the last inequality is by Cauchy--Schwarz, and we are defining $\ket{F_z} = \sum_{w}(-1)^{z\cdot w}\ket{f_w}$, and $\ket{G_z} = \sum_{w}(-1)^{z\cdot w}c^{\perp}_w \ket{f_w}$. We will now leverage the fact that our robustness guarantee of Theorem~\ref{thm:1} is ``with high probability over a Haar random target state''. We establish two lemmas that leverage the Haar measure's concentration to obtain high probability bounds on $\| \ket{\hat{f}_z}\|^2$ and $\|\ket{F_z} \|^2$.
\begin{lem}
\label{lem:3}
With probability $1 -O(2^{-n})$ over sampling a Haar random $n$-qubit target pure state $\ket{\psi} = \sum_{z} \ket{z} \ket{f_z}$, the following hold for all $z\in \{0,1\}^r$ (where recall that $r = n-2\log n$):
\begin{itemize}
    \item $\| \ket{\hat{f}_z}\|^2 = (1 \pm o(1))\cdot \frac{1}{2^r}$, where $\ket{\hat{f}_z} = \frac{1}{\sqrt{2^r}}\sum_{w \in \{0,1\}^r} (-1)^{z\cdot w} \ket{\widehat{f_w}}$. 
    \item $\|\ket{F_z} \|^2 = (1 \pm o(1))$, where $\ket{F_z} = \sum_{w}(-1)^{z\cdot w}\ket{f_w}$.
\end{itemize}
\end{lem}
\begin{proof}
These are somewhat standard bounds, but we provide a proof for completeness. The distribution of the vectors $\{\ket{f_z}\}$ is identical to the following: sample i.i.d. Gaussian vectors, and normalize so that their squared norms sum to $1$. More precisely, \ $\ket{f_z} \sim \frac{\ket{g_z}}{\| G\|_F }$, where $G$ is a $2^{n-r} \times 2^r$ matrix with i.i.d. $\mathcal{CN}(0,1)$ entries (\textit{i.e.}\ standard complex Gaussian), and the $\ket{g_z}$ are the columns of $G$. So, $\| f_z\|^2 = \frac{\| \ket{g_z}\|^2}{\|G\|_F^2}$. Now, $\| G\|_F^2$ is the sum of the magnitudes squared of $2^n$ i.i.d.\ $\mathcal{CN}(0,1)$ random variables. Thus, its distribution is identical to the sum of $2^n$ $\textnormal{Exp}(1)$ random variables (which has mean $2^n$). We can invoke a standard Chernoff bound for sub-exponential random variables to deduce that
$$ \Pr\Big[\big\| G\|^2_F - 2^n \big| > \epsilon \cdot 2^n\Big] < e^{\frac{-2^n \epsilon^2}{3}} \,.$$
Similarly, recalling that $2^{n-r} = n^2$, we have that, for any $z$,
$$ \Pr\Big[\big | \| \ket{g_z}\|^2 - n^2 \big| > \epsilon \cdot n^2\Big] < e^{\frac{-n^2 \epsilon^2}{3}}\,.$$

Taking $\epsilon = n^{\frac14}$, the RHS becomes $e^{-\frac{2^{3n/2}}{3}}$. This is enough to enable a union bound over all $z \in \{0,1\}^{r}$, and establish that, with probability at least $1-O(2^{-n})$, $\| G\|^2_F =  (1\pm o(1)) 2^n$, and $\| \ket{g_z}\|^2  = (1 \pm o(1)) n^2$ for all $z$. This implies that, with the same probability, for all $z$, 
$$ \| \ket{f_z} \|^2 = (1\pm o(1)) \frac{n^2}{2^n} = (1\pm o(1)) \frac{1}{2^r} \,.$$
The vectors $\ket{\hat{f}_z}$ in the lemma statement satisfy the exact same bound since $\ket{\hat{f}} = H^{\otimes n} \ket{f}$ is also distributed as a Haar random state, by the unitary invariance of the Haar measure.

To bound $\|\ket{F_z}\|^2$, note first that $\ket{F_z}$ is distributed identically to $\frac{1}{\|G\|_F}\sum_{w} (-1)^{z \cdot w}\ket{g_w} $, where $G$ and $\ket{g_w}$ are as before. Let $\ket{\tilde{g}_z} = \frac{1}{\sqrt{2^r}} \sum_w (-1)^{z\cdot w} \ket{g_w}$. Then, we claim that the $\ket{\tilde{g}_z}$ are i.i.d.\ vectors with i.i.d.\ $\mathcal{CN}(0,1)$ entries. This follows by unitary invariance, by considering the linear map on $\mathbb{C}^{2^r \times n^2}$ that sends the vector obtained by ``stacking''
the $\ket{g_z}$ one after the other, to the vector obtained by stacking the $\ket{\tilde{g}_z}$. One can easily verify that this map is unitary. Thus, the latter vector also has i.i.d.\ $\mathcal{CN}(0,1)$ entries, and so does each $\ket{\tilde{g}_z}$. So $\ket{F_z}$ is distributed as $\frac{\sqrt{2^r}}{\|G\|_F} \ket{\tilde{g}_z}$, where the $\tilde{g}_z$ have i.i.d.\ $\mathcal{CN}(0,1)$ entries. Invoking the same bounds obtained previously, we have that, with probability at least $1-O(2^{-n})$, 
$$ \| \ket{F_z} \|^2  = (1 \pm o(1)) 2^r \cdot \frac{1}{2^n} \cdot n^2 = (1 \pm o(1)) \,,$$
as desired.
\end{proof}

\medskip

Applying Lemma~\ref{lem:3} to \eqref{eq:13}, we have that, with probability $1 -O(2^{-n})$ over sampling a Haar random target state, 
\begin{align}
    \eqref{eq:8} &\leq \frac{(1+o(1))}{2^r} \sum_{z \in \{0,1\}^r} \| \ket{G_z} \|^2  \nonumber\\ &= \frac{(1+o(1))}{2^r} \sum_{z \in \{0,1\}^r} \| \sum_w (-1)^{z\cdot w} c^{\perp}_w \ket{f^{\perp}_w} \|^2 \nonumber \\ 
    &= \frac{(1+o(1))}{2^r} \cdot 2^r \sum_{w \in \{0,1\}^r} \| c^{\perp}_w\ket{f^{\perp}_w}\|^2 \nonumber\\
    &= (1+o(1))\sum_{w \in \{0,1\}^r} |c^{\perp}_w|^2 \label{eq:155}
\end{align}
where the second equality is due to Parseval's identity for vector-valued functions. 

We now proceed to bounding term \eqref{eq:7}. We will spend almost the entire rest of the proof doing this. The key bounds that we will prove along the way will then also allow us to bound term \eqref{eq:6} with little extra work. 

We start by simplifying the expression.
\begin{align}
     \eqref{eq:7} &= \frac{1}{2^{2r}} \sz \frac{1}{\|\ket{\hat{f}_z}\|^2} \cdot \sum_{\substack{w\neq w' \\ v\neq v'}} (-1)^{z \cdot (w\oplus w' \oplus v \oplus v')}  2 \textnormal{Re}(\overline{c_w}c_v^{\perp} \overline{A_{w,w'}}\cdot B_{v,v'}) \nonumber\\
     &= \frac{1}{2^{2r}} \sz \frac{1}{\|\ket{\hat{f}_z}\|^2} \cdot 2 \textnormal{Re} \left(\sum_{w\neq w'}(-1)^{z \cdot (w\oplus w')}\overline{c_w}  \overline{A_{w,w'}}  \cdot \sum_{v\neq v'} (-1)^{z \cdot (v \oplus v')} c_v^{\perp} B_{v,v'}\right) \nonumber \\
     & \leq  \frac{1}{2^{2r}} \sz \frac{1}{\|\ket{\hat{f}_z}\|^2} \cdot 2 \left|\sum_{w\neq w'}(-1)^{z \cdot (w\oplus w')}\overline{c_w}  \overline{A_{w,w'}}\right| \cdot \left| \sum_{v\neq v'} (-1)^{z \cdot (v \oplus v')} c_v^{\perp} B_{v,v'} \right| \nonumber\\
     &\leq \frac{2+ o(1)}{2^r} \sz \left|\sum_{w\neq w'}(-1)^{z \cdot (w\oplus w')}\overline{c_w}  \overline{A_{w,w'}}\right| \cdot \left| \sum_{v\neq v'} (-1)^{z \cdot (v \oplus v')} c_v^{\perp} B_{v,v'} \right| \nonumber \\
     &\leq \frac{2+ o(1)}{2^r} \sqrt{\sum_z \Bigg|\sum_{w\neq w'}(-1)^{z \cdot (w\oplus w')}\overline{c_w}  \overline{A_{w,w'}}\Bigg|^2} \cdot \sqrt{\sum_z \Bigg| \sum_{v\neq v'} (-1)^{z \cdot (v \oplus v')} c_v^{\perp} B_{v,v'} \Bigg|^2}\,, \label{eq:15}
\end{align}
where, by Lemma~\ref{lem:3}, the second inequality holds with probability $1 -O(2^{-n})$ over sampling a Haar random target state; and the third inequality is by Cauchy--Schwarz. From here on, our bounds on \eqref{eq:7} will hold only with high probability. Now, notice that, combining \eqref{eq:99} with the guarantee from Lemma~\ref{lem:3}, we have that, with probability $1 -O(2^{-n})$,
\begin{equation*}
 \sum_z \bigg| \sum_{v\neq v'} (-1)^{z \cdot (v \oplus v')} c_v^{\perp} B_{v,v'} \bigg|^2 \leq (1\pm o(1)) \cdot 2^r \cdot \eqref{eq:8} \,,
\end{equation*}
But, we have already previously obtained a high probability bound on \eqref{eq:8} in Equation~\eqref{eq:155}. Using this, we have
\begin{equation*}
 \sum_z \bigg| \sum_{v\neq v'} (-1)^{z \cdot (v \oplus v')} c_v^{\perp} B_{v,v'} \bigg|^2 \leq (1+o(1)) \cdot 2^r \sum_{w} |c^{\perp}_w|^2 
\end{equation*}
(overall, in our entire proof, our uses of Lemma~\ref{lem:3} can be accounted for by a union bound which does not affect the $1 -O(2^{-n})$ probability). Denoting $| c^{\perp} | := \sqrt{\sum_w |c^{\perp}_w|^2}$, and substituting this into \eqref{eq:15}, we get
\begin{align}
     \eqref{eq:7} &\leq \frac{2+o(1)}{\sqrt{2^r}}\cdot |c^{\perp}| \cdot \sqrt{\sum_z \Bigg|\sum_{w\neq w'}(-1)^{z \cdot (w\oplus w')}\overline{c_w}  \overline{A_{w,w'}}\Bigg|^2} \label{eq:17}
\end{align}
Now, focusing on the term inside the square root, and remembering the definition of $A_{w,w'}$, we have
\begin{align}
    &\sum_z \Bigg|\sum_{w\neq w'}(-1)^{z \cdot (w\oplus w')}\overline{c_w}  \overline{A_{w,w'}}\Bigg|^2 \nonumber \\ 
    = & \sum_z \Bigg|\sum_{w\neq w'}(-1)^{z \cdot (w\oplus w')}\overline{c_w}  \frac{\braket{f_{w'}|f_w}}{\| f_w\|} \Bigg|^2 \nonumber \\
    =& \sum_z \sum_{\substack{ w \neq w' \\ v \neq v' }} (-1)^{z \cdot ( w\oplus w'\oplus v \oplus v')} \overline{c_w} c_v \cdot \frac{\braket{f_{w}|f_{w'}}}{\| f_w\|}  \cdot \frac{\braket{f_{v'}|f_v}}{\| f_v\|} \nonumber \\
    =& \, 2^r \cdot \sum_{\substack{ w \neq w' \\ v \neq v': \\ w \oplus w'\oplus v \oplus v' =0^r}} \overline{c_w} c_v\frac{\braket{f_{w}|f_{w'}}}{\| f_w\|}   \cdot \frac{\braket{f_{v'}|f_v}}{\| f_v\|} \nonumber \\
     =& \, 2^r \cdot \sum_{w,v} \overline{c_w} c_v M_{wv} \,,
\end{align}
where $M$ is $2^r \times 2^r$ matrix with entries $$M_{wv} = \sum_{\substack{ w'\neq w \\ v' \neq v: \\ w \oplus w'\oplus v \oplus v' =0^r}}  \frac{\braket{f_{w}|f_{w'}}}{\| f_w\|} \frac{\braket{f_{v \oplus w \oplus w'}|f_v}}{\| f_v\|} \,. $$
Letting $|c|$ denote the norm of the vector with entries $c_w$, and noting that $|c| = \sqrt{1 - |c^{\perp}|^2}$, we have
\begin{equation}
\label{eq:200}
\sum_z \Bigg|\sum_{w\neq w'}(-1)^{z \cdot (w\oplus w')}\overline{c_w}  \overline{A_{w,w'}}\Bigg|^2 \leq 2^r (1-|c^{\perp}|^2) \cdot \| M \| \,.
\end{equation}
Substituting the latter bound into \eqref{eq:17}, we have 
\begin{equation}
     \eqref{eq:7} \leq (2+o(1)) \cdot |c^{\perp}| \sqrt{1-|c^{\perp}|^2} \cdot \sqrt{\| M \|}  \,. \label{eq:22}
\end{equation}
Now, notice that, when the target state is sampled from the Haar measure, the distribution from which the $(n-r)$-qubit vectors $\ket{f_w}$ are sampled is identical to the following: sample each of their entries as i.i.d.\ $\mathcal{CN}(0,1)$ (\textit{i.e.}\ standard complex Gaussian) entries, and then normalize so that the squared norms add up to~1. Equivalently, sample a $2^{n-r} \times 2^{r}$ matrix $G$ with i.i.d.\ $\mathcal{CN}(0,1)$ entries; let $\ket{g_w}$ denote the column indexed by $w$, and let $\| G\|_F$ be the Frobenius norm of $G$; let $\ket{f_w} = \frac{\ket{g_w}}{\|G\|_F}$. Then, we can rewrite the entries of $M$ as 
$$M_{wv} = \frac{1}{\| G\|^2_F} \sum_{\substack{ w'\neq w \\ v' \neq v: \\ w \oplus w'\oplus v \oplus v' =0^r}}   \frac{\braket{g_w|g_{w'}}}{\| g_w\|} \frac{\braket{g_{v \oplus w \oplus w'}|g_v}}{\| g_v\|} \,, $$
and our goal is to obtain a high probability bound on the operator norm of $M$. We will now make use of the following lemma about Gaussian concentration.
\begin{lem}
\label{lem:gaussian-matrix}
Let $\epsilon>0$ and $d,d'\in \mathbb{N}$. Let $G$ be a $d \times d'$ matrix with i.i.d.\ $\mathcal{CN}(0,1)$ (\textit{i.e.}\ standard complex Gaussian) entries. Then, 
\begin{itemize}
\item $\Pr[\|G\|_F^2 < (1-\epsilon)dd'] < 2^{-\frac{\epsilon^2dd'}{4}}$. And, consequently, $\Pr[\|G\|_F < (1-\epsilon)\sqrt{dd'}] < 2^{-\frac{\epsilon^2dd'}{8}}$.
\item Similarly, for any $w\in [d']$, letting $\ket{g_w}$ denote the $w$-th column of $G$, we have that,
$$\Pr\Big[\|\ket{g_w}\|^2 \notin [(1-\epsilon)d, (1+\epsilon)d]\Big] < 2\cdot 2^{-\frac{\epsilon^2d}{4}}\,.$$ 
\end{itemize}
\end{lem}
\begin{proof}
$\|G\|_F^2$ is a sum of $dd'$ i.i.d.\ $\textnormal{Exp}(1)$ random variables. This sum concentrates strongly around the mean, which is $dd'$. A standard Chernoff bound for sub-exponential random variables gives the desired bound. Likewise, $\| \ket{g_w}\|^2$ is a sum of $d$ i.i.d.\ $\textnormal{Exp}(1)$ random variables. The same Chernoff bound gives the desired bound.
\end{proof}
In our setting, $d d' = 2^n$. Taking $\epsilon = 2^{n/4}$, gives that, except with doubly exponentially small probability in $n$, $\|G\|_F >(1-o(1))\sqrt{2^n}$. 
Now, define the matrix $\widetilde{M} = \| G\|_F^2 \cdot M$, \textit{i.e.}\ the matrix with entries 
\begin{equation}
\label{eq:m-tilde}
\widetilde{M}_{wv}=  \sum_{\substack{ w'\neq w \\ v' \neq v: \\ w \oplus w'\oplus v \oplus v' =0^r}}  \frac{\braket{g_w|g_{w'}}}{\| g_w\|} \frac{\braket{g_{v \oplus w \oplus w'}|g_v}}{\| g_v\|} \,.
\end{equation}
Then, substituting this into \eqref{eq:22}, gives 
\begin{equation*}
     \eqref{eq:7} \leq (2+o(1)) \cdot |c^{\perp}| \sqrt{1-|c^{\perp}|^2} \cdot \frac{1}{\| G\|_F}\sqrt{\| \widetilde{M} \|}  \,.
\end{equation*}
And given what we just derived, we have that, except with doubly exponentially small probability in $n$ (and merging the two terms that contain $o(1)$ terms),
\begin{equation}
     \eqref{eq:7} \leq \frac{(2+o(1))}{\sqrt{2^n}} \cdot |c^{\perp}| \sqrt{1-|c^{\perp}|^2} \cdot \sqrt{\| \widetilde{M} \|}  \,. \label{eq:23}
\end{equation}
Our goal is now to show that $\| \widetilde{M}\| = O(2^n)$. This is sufficient to conclude the proof of Theorem~\ref{thm:1}, as we do at the end, starting from Equation~\eqref{eq:end}.

The first key observation to bound $\| \widetilde{M}\|$ is that $\widetilde{M} = E \cdot E^{\dagger}$, where $E$ is the $2^r \times (2^r-1)$ matrix with the following entries, for $w \in \{0,1\}^r$, $s \in \{0,1\}^r\setminus\{0^r\}$:
$$E_{ws} = \frac{\braket{g_w| g_{w\oplus s}}}{\| \ket{g_w}\|} \,.$$
To see this, notice that
\begin{align}
\widetilde{M}_{w,v} &= \sum_{\substack{ w'\neq w \\ v' \neq v: \\ w \oplus w'\oplus v \oplus v' =0^r}}   \frac{\braket{g_w|g_{w'}}}{\| g_w\|} \frac{\braket{g_{v \oplus w \oplus w'}|g_v}}{\| g_v\|} \\
&= \sum_{s\neq 0^r} \frac{\braket{g_w|g_{w \oplus s}}}{\| g_w\|} \frac{\braket{g_{v \oplus s}|g_v}}{\| g_v\|} \\
&= \sum_{s\neq 0^r} E_{ws} \cdot E^{\dagger}_{sv} \,,
\end{align}
where the first equality is by changing variables in the summation setting $s = w'\oplus w$. 

Now, since $\| \widetilde{M}\| = \| E\|^2$, our goal of showing that $\|\widetilde{M}\| = O(2^n)$ is equivalent to showing that $ \| E\| = O(\sqrt{2^n})$. Thus, in the rest of the proof, we will focus on showing the following lemma. 
\begin{lem}
\label{lem:5}
Let $r = n - 2\log n$. Let $G \in \mathbb{C}^{2^{n-r} \times 2^r}$ be a matrix with i.i.d. $\mathcal{CN}(0,1)$ entries. For $w \in \{0,1\}^r$, let $\ket{g_w} \in \mathbb{C}^{2^{n-r}}$ be the column of $G$ indexed by $w$ (when indexing columns of $G$ by $r$-bit strings). Let $E \in \mathbb{C}^{2^r \times (2^r - 1)}$ be the matrix with the following entries, for $w \in \{0,1\}^r$, $s \in \{0,1\}^r\setminus\{0^r\}$:
$$E_{ws} = \frac{\braket{g_w| g_{w\oplus s}}}{\| \ket{g_w}\|} \,.$$
Then, with probability $1 -O(2^{-n})$, 
$$ \| E \| \leq  \sqrt{2^n} \,.$$
\end{lem}
To prove Lemma~\ref{lem:5}, we will leverage a key result by de la Pe\~{n}a and Montgomery-Smith~\cite{de1995decoupling} that allows us to ``decouple'' the two Gaussian random vectors that appear in the definition of $E_{ws}$. We state the general version of the result, and then the specialization that applies to our setting.
\begin{lem}[\cite{de1995decoupling}]
\label{lem:decoupling}
Let $X_i$ be a sequence of independent random variables taking values in a measure space $S$, and let $f_{i_1 \ldots i_k}$, for $i_1, \ldots, i_k \in [m]$, be measurable functions from $S^k$ to a Banach space $B$. Let $(X_i^{(j)})$ be independent copies of $(X_i)$. The following inequality holds for all $t\geq 0$ and all $m \geq2$, for some universal constant $C_k$ that only depends on $k$.
$$\Pr\Bigg[\bigg\| \sum_{1\leq  i_1 \neq \cdots \neq i_m \leq m} f_{i_1\ldots i_k}(X_{i_1}, \ldots, X_{i_k}) \bigg\| \geq t\Bigg] \leq C_k \cdot \Pr\Bigg[C_k \cdot\bigg\|\sum_{1\leq  i_1 \neq \cdots \neq i_m \leq m}  f_{i_1\ldots i_k}(X_{i_1}^{(1)}, \ldots, X_{i_k}^{(k)}) \bigg \| \geq t \Bigg] \,.$$
\end{lem}
\begin{cor}
\label{cor:10}
Let $E \in \mathbb{C}^{2^r \times (2^r - 1)}$ be the random matrix defined in Lemma~\ref{lem:5}. Let $\tilde{E}  \in \mathbb{C}^{2^r \times (2^r - 1)}$ be the random matrix with entries:
$$\tilde{E}_{ws} = \frac{\braket{g_w| \tilde{g}_{w\oplus s}}}{\| \ket{g_w}\|} \,,$$
where the only difference from $E$ is that $\{\ket{\tilde{g}_w}\}$ is an independent set of vectors with the same distribution as $\{\ket{g_w}\}$ (and recall that each $\ket{g_w}$ has i.i.d.\ $\mathcal{CN}(0,1)$ entries, and the vectors themselves are i.i.d). Then, for any $t \geq 0$,
$$\Pr\Big[ \| E\| \geq t \| \Big] \leq C \cdot \Pr\Big[ C \cdot \| \tilde{E} \|\geq t \Big] \,,$$
where $C$ is the constant $C_2$ from Lemma~\ref{lem:decoupling}.
\end{cor}
\begin{proof}
Apply Lemma~\ref{lem:decoupling} with $k=2$, and the following settings: let the Banach space $B$ be $\mathbb{C}^{2^r \times (2^r - 1)}$ equipped with the operator norm; let $S$ be $\mathbb{C}^{2^{n-r}}$ with Euclidean norm (any norm is fine here); denote the indices $i_1, i_2$ by $w, w'$, and take them to range over $\{0,1\}^{r}$; let $X_w = \ket{g_w}$; let $f_{ww'}(\ket{g_w},\ket{g_{w'}})$ be the $2^r \times (2^r - 1)$ matrix that is zero everywhere except at the $(w,w')$ entry, where it is equal to $\frac{\braket{g_w| \tilde{g}_{w'}}}{\| \ket{g_w}\|}$. Then, notice that 
$$ \sum_{w \neq w'} f_{ww'}(\ket{g_w}, \ket{g_{w'}}) = E \,,$$
where, crucially, the summation is only over $w, w'$ such that $w\neq w'$. Similarly, if we consider a second independent set of random variables $\{\tilde{g}_w\}$, we have
\[ \sum_{w \neq w'} f_{ww'}(\ket{g_w}, \ket{\tilde{g}_{w'}}) = \tilde{E} \,.\qedhere\]
\end{proof}
So, from here on, our goal is to show that, there is a constant $C'$ such that, with overwhelming probability over sampling the vectors $\{\ket{g_w}\}$ and $\{\ket{\tilde{g}_w}\}$, $\| \tilde{E} \| \leq C' \sqrt{2^n} $, where $\tilde{E} \in \mathbb{C}^{2^{r} \times (2^{r}-1)}$ is the matrix with entries
$$\tilde{E}_{ws} = \frac{\braket{g_w| \tilde{g}_{w\oplus s}}}{\| \ket{g_w}\|} \,.$$

Our strategy will be to put $\tilde{E}$ in a form that is amenable to a standard matrix Chernoff bound (this step will leverage the ``decoupled'' form of $\tilde{E}$).

For $j \in \{0,1\}^{n-r}$, let $g^{(j)}_w \in \mathbb{C}$ denote the $j$-th entry of $\ket{g_w}$ (here we are taking the entries to be indexed by strings). Then, note that $\{g^{(j)}_w\}$ is a collection of i.i.d.\ $\mathcal{CN}(0,1)$ random variables. We use the analogous notation for $\{\tilde{g}^{(j)}_w\}$. The first key observation is that we can write 
\begin{equation}
\label{eq:tildeE}
\tilde{E} = \sum_{\substack{w \in \{0,1\}^r\\ j\in \{0,1\}^{n-r}}} \tilde{g}^{(j)}_w \cdot B^{wj} \,,
\end{equation}
where each $B^{wj} \in \mathbb{C}^{2^r \times (2^r-1)}$ is a matrix that is zero everywhere except for the $w$-th row: for $s \in \{0,1\}^r \setminus \{0^r\}$, the $s$-th entry of the $w$-th row is $\overline{g^{(j)}_{w\oplus s}}/\| \ket{g_w}\|$. We are now ready to apply the following standard matrix Chernoff bound.
\begin{lem}
\label{lem:matrix-chernoff}
    Let $m \in \mathbb{N}$. Let $\{B_k : k \in [m]\}$ be a collection of complex matrices with dimension $d_1 \times d_2$, and let $\{\xi_k: k \in [m]\}$ be a collection of i.i.d.\ $\mathcal{CN}(0,1)$. Let $$\sigma^2 = \max \left\{\Big\|\sum_k B_kB_k^{\dagger}\Big\|, \Big\|\sum_k B_k^{\dagger}B_k\Big\|\right\} \,.$$
Then, for all $t \geq 0$,
$$ \Pr\bigg[\Big\| \sum_k \xi_k B_k \Big\| \geq t \bigg] \leq (d_1+d_2) \cdot e^{-t^2/2\sigma^2} \,.$$
\end{lem}
We can apply Lemma~\ref{lem:matrix-chernoff} to our matrix $\tilde{E}$ to get
\begin{equation}
\label{eq:266}
\Pr[\| \tilde{E}\| \geq t] \leq (2^{r+1}-1) \cdot e^{-t^2/2\sigma^2} \,,
\end{equation}
where 
\begin{equation}
\label{eq:sigma}
\sigma^2 = \max \left\{\Bigg\| \sum_{\substack{w \in \{0,1\}^r\\ j\in \{0,1\}^{n-r}}} B^{wj} \cdot (B^{wj})^{\dagger} \Bigg\|, \Bigg\|\sum_{\substack{w \in \{0,1\}^r\\ j\in \{0,1\}^{n-r}}}(B^{wj})^{\dagger}\cdot B^{wj} \Bigg\|\right\} \,,
\end{equation}
and the matrices $B^{wj}$ are defined right after Equation~\eqref{eq:tildeE}.

We will now bound each of the two operator norms inside the maximization. Observe first that $$\sum_{\substack{w \in \{0,1\}^r\\ j\in \{0,1\}^{n-r}}} B^{wj} \cdot (B^{wj})^{\dagger} = \sum_{w \in \{0,1\}^r} \frac{1}{\|\ket{g_w} \|^2} \sum_{j\in \{0,1\}^{n-r}} \widetilde{B^{wj}} \cdot (\widetilde{B^{wj}})^{\dagger} \,,$$
where $\widetilde{B^{wj}} = \| \ket{g_w}\|\cdot B^{wj}$, \textit{i.e.}\ $\widetilde{B^{wj}}  \in \mathbb{C}^{2^r \times (2^r-1)}$ is a matrix that is zero everywhere except for the $w$-th row: for $s \in \{0,1\}^r \setminus \{0^r\}$, the $s$-th entry of the $w$-th row is $\overline{g^{(j)}_{w\oplus s}}$. Now, by Lemma~\ref{lem:gaussian-matrix}, taking $\epsilon = 0.1$, and $d = 2^{n-r} = 2^{2\log n} = n^2$, we have that, for any $w \in \{0,1\}^r$, 
$$\Pr\Big[\|\ket{g_w}\|^2 \notin [0.9d, 1.1 d]\Big] < 2\cdot 2^{-\frac{0.01 n^2}{4}}\,.$$ 
By a union bound, we have that, with probability at least $1- 2\cdot2^{-\frac{0.01 n^2}{4}+r}$, the above holds for all $w$. Now, crucially, each term $\widetilde{B^{wj}} \cdot (\widetilde{B^{wj}}^{\dagger})$ is positive semi-definite. Thus, it follows that, with probability at least  $1- 2\cdot 2^{-\frac{0.01 n^2}{4}+r}$, 
\begin{equation}
\label{eq:345}
\Bigg\| \sum_{\substack{w \in \{0,1\}^r\\ j\in \{0,1\}^{n-r}}} B^{wj} \cdot (B^{wj})^{\dagger} \Bigg\| \leq \frac{1.12}{d} \cdot \Bigg\| \sum_{\substack{w \in \{0,1\}^r\\ j\in \{0,1\}^{n-r}}} \widetilde{B^{wj}} \cdot (\widetilde{B^{wj}})^{\dagger} \Bigg\|\,.
\end{equation}

Now, notice that $\widetilde{B^{wj}} \cdot (\widetilde{B^{wj}})^{\dagger}$ is particularly simple: it is a $2^r \times 2^r$ matrix that is zero everywhere except at the $(w,w)$ entry, where it is equal to $\sum_{s \in \{0,1\}^r\setminus\{0^r\}} |g^{(j)}_{w\oplus s}|^2$. Thus, the matrix on the RHS of \eqref{eq:345} is diagonal with $(w,w)$ entry $\sum_{s \in \{0,1\}^r\setminus\{0^r\}} \| \ket{g_{w\oplus s}}\|^2 = \sum_{w' \neq w} \| \ket{g_{w'}}\|^2 \leq 2^r \cdot 1.1 d$, which, with probability at least $1- 2\cdot 2^{-\frac{0.01 n^2}{4}+r}$, holds for all $w$. So, with the same probability,
\begin{equation}
\Bigg\| \sum_{\substack{w \in \{0,1\}^r\\ j\in \{0,1\}^{n-r}}} B^{wj} \cdot (B^{wj})^{\dagger} \Bigg\| \leq \frac{1.12}{d} \Bigg\| \sum_{\substack{w \in \{0,1\}^r\\ j\in \{0,1\}^{n-r}}} \widetilde{B^{wj}} \cdot (\widetilde{B^{wj}})^{\dagger} \Bigg\| \leq \frac{1.12}{d} \cdot 2^r \cdot 1.1d \leq 1.3 \cdot 2^r\,.
\end{equation}
As we will see soon enough, this bound is sufficient for us to use in Equation~\eqref{eq:266}. So, our final task is to show a similar bound on the \emph{other} operator norm in the definition of $\sigma^2$, namely
\begin{equation}
\label{eq:267}
\Bigg\|\sum_{\substack{w \in \{0,1\}^r\\ j\in \{0,1\}^{n-r}}}(B^{wj})^{\dagger}\cdot B^{wj} \Bigg\| \,.
\end{equation}
This will require a bit more work. First, analogously to \eqref{eq:345}, with probability at least $1- 2\cdot 2^{-\frac{0.01 n^2}{4}+n}$,
\begin{equation}
\label{eq:346}
\Bigg\|\sum_{\substack{w \in \{0,1\}^r\\ j\in \{0,1\}^{n-r}}}(B^{wj})^{\dagger}\cdot B^{wj} \Bigg\| \leq \frac{1.12}{d} \Bigg\|\sum_{\substack{w \in \{0,1\}^r\\ j\in \{0,1\}^{n-r}}}(\widetilde{B^{wj}})^{\dagger}\cdot \widetilde{B^{wj}} \Bigg\| \,,
\end{equation}
where $d = 2^{n-r}$.

Fix $w \in \{0,1\}^r$ and $j \in \{0,1\}^{n-r}$. Then, notice that $(\widetilde{B^{wj}})^{\dagger}\cdot \widetilde{B^{wj}}$ is $2^{r-1} \times 2^{r-1}$ matrix with $(i,k)$ entry: $g^{(j)}_{w\oplus i}\cdot \overline{g^{(j)}_{w\oplus k}}$. Now, fixing $j$ and summing over $w$ yields the matrix $$C^{(j)} = \sum_{w\in \{0,1\}^r} (\widetilde{B^{wj}})^{\dagger}\cdot \widetilde{B^{wj}}$$ with $(i,k)$ entry: 
$$\sum_{w\in \{0,1\}^r}  g^{(j)}_{w\oplus i}\cdot \overline{g^{(j)}_{w\oplus k}} = \sum_{w\in \{0,1\}^r} g^{(j)}_{w}\cdot \overline{g^{(j)}_{w\oplus s}} \,,$$ where $s = i\oplus k$. Thus, crucially, $C^{(j)}$ is a \emph{circulant} matrix, \textit{i.e.}\ the $(i,k)$ entry only depends on $i\oplus k$ (formally, it is circulant with respect to $\mathbb{Z}_2^n$).
Since $C^{(j)}$ is circulant, it is diagonalized by the Walsh-Hadamard transform (this is formalized in the lemma below)! Crucially, the operator whose operator norm we want to bound (on the RHS of~\eqref{eq:346}) is $C = \sum_{j \in \{0,1\}^{n-r}} C^{(j)}$, and the $C^{(j)}$ are all diagonal in the same basis. So, all we need to do is find the eigenvalues, and understand the distribution of their sums. The following lemma characterizes the eigenvalues and eigenvectors of circulant matrices.
\begin{lem}
\label{lem:circulant}
Let $m \in \mathbb{N}$. Let $W \in \mathbb{C}^{2^m \times 2^m}$, with entries indexed by $m$-bit strings. Suppose $W$ is circulant, \textit{i.e.}\ there is a function $\alpha(\cdot)$ such that, for all $i,k \in \{0,1\}^m$, $W(i,k) = \alpha(s)$, where $s = i\oplus k$. Then, $W$ has eigenvectors $v_t$, for $t \in \{0,1\}^m$, where the $x$-th entry of $v_t$ is
$$ v_t(x) = (-1)^{t \cdot x} \,.$$
Its corresponding eigenvalue is $\lambda_t = \sum_{s \in \{0,1\}^m} \alpha(s) (-1)^{t\cdot s}$.
\end{lem}
\begin{proof}
We directly verify that, for any $y \in \{0,1\}^m$,
$$ (W v_t) (i) = \sum_k W(i,k) v_t(k) = \sum_s \alpha(s) (-1)^{t \cdot (i \oplus s)} = (-1)^{t\cdot i} \sum_s \alpha(s) (-1)^{t\cdot s} = v_t(i) \cdot \lambda_t\,,$$
where $\lambda_t$ is defined as in the lemma statement. So, $W v_t = \lambda_t v_t$.
\end{proof}
Going back to $C^{(j)}$, we can apply Lemma~\ref{lem:circulant} with $\alpha(s) = \sum_{w\in \{0,1\}^r} g^{(j)}_{w}\cdot \overline{g^{(j)}_{w\oplus s}}$, for $s \in \{0,1\}^r$, to deduce that the eigenvalues of $C^{(j)}$ are \begin{align*}
\lambda^{(j)}_t &= \sum_s \alpha(s) (-1)^{t\cdot s} \\
&= \sum_{w,s} g^{(j)}_{w}\cdot \overline{g^{(j)}_{w\oplus s}} (-1)^{t\cdot s}  \\ 
&=  \sum_{w,w'} g^{(j)}_{w}\cdot \overline{g^{(j)}_{w'}} (-1)^{t\cdot (w\oplus w')} \\
&=  \left(\sum_w g^{(j)}_{w} (-1)^{t\cdot w} \right) \left(\sum_{w'} \overline{g^{(j)}_{w'}} (-1)^{t\cdot w'}\right) \\
&= \left| \sum_w g^{(j)}_{w} (-1)^{t\cdot w}    \right|^2
\end{align*}
Now, let $g^{(j)} \in \mathbb{C}^{2^r}$ denote the vector with entries $g^{(j)}_w$. Then, notice that
$$ \lambda_t^{(j)} = 2^r \big| (H g^{(j)})_t \big|^2 \,,$$
where $H$ is the Walsh-Hadamard transform (which can be equivalently viewed as the tensor product of $r$ ``single-qubit Hadamard'' matrices) - this is the $2^r \times2^r$ matrix with entries $H_{ww'} = \frac{(-1)^{w\cdot w'}}{\sqrt{2^r}}$. 

Now, $g^{(j)}$ is a vector with i.i.d.\ $\mathcal{CN}(0,1)$ entries. By unitary invariance, $H g^{(j)}$ has the same distribution. Thus, the $\lambda^{(j)}_t/2^r$ are i.i.d.\ $\textnormal{Exp}(1)$ random variables. 

Recall that we wish to bound the operator norm of $C =  \sum_{j \in \{0,1\}^{n-r}} C^{(j)}$. Crucially, all of the $C^{(j)}$ are diagonalized by the Walsh-Hadamard transform, \textit{i.e.}\ they have the same set of eigenvectors from Lemma~\ref{lem:circulant}. Thus, the eigenvalues of $C$ are simply given by
$$ \lambda_t = \sum_{j \in \{0,1\}^{n-r}} \lambda^{(j)}_t \textnormal{ for } t \in \{0,1\}^r \,.$$ 
For a fixed $t \in \{0,1\}^r$, we can now leverage the fact that $\lambda_t/2^r$ is a sum of i.i.d.\ Exp(1) random variables, and invoke a standard Chernoff to get that, for any $\epsilon >0$,
$$ \Pr[\lambda_t/2^r \geq (1+\epsilon) \cdot 2^{n-r} ] < e^{-2^{n-r}  \epsilon^2/3} \,.$$
Taking $\epsilon = 0.1$ and rearranging, we have
$$ \Pr[\lambda_t \geq 1.1 \cdot 2^n ] < e^{-2^{n-r} 0.01/3} \,.$$
Recall that $n-r = 2\log n$. So, by a union bound, with probability at least $1-2^r \cdot e^{\frac{0.01n^2}{3}}$, we have that, for all $t$,
$$\lambda_t \leq 1.1\cdot 2^n \,,$$
which implies that $\|C\| \leq 1.1\cdot 2^n$ with the same probability. Recall that $\|C\|$ is the operator norm appearing on the RHS of \eqref{eq:346}. So, substituting there, we have that 
\begin{equation}
\label{eq:B2}
\Bigg\|\sum_{\substack{w \in \{0,1\}^r\\ j\in \{0,1\}^{n-r}}}(B^{wj})^{\dagger}\cdot B^{wj} \Bigg\| \leq \frac{1.12}{d} \|C\| \leq \frac{1.12}{d} \cdot 1.1\cdot 2^n \leq 1.3 \cdot 2^r\,,
\end{equation}
where recall that $d = 2^{n-r} = n^2$. 

We are finally ready to obtain a bound on $\sigma^2$ from \eqref{eq:sigma}. Substituting \eqref{eq:345} and \eqref{eq:B2} into \eqref{eq:sigma}, we get that, with probability at least $1- 2\cdot 2^{-\frac{0.01 n^2}{4}+r} -  e^{\frac{0.01 n^2}{3}+r}$, where recall that this probability is over sampling the vectors $\{\ket{g_w}\}$, we have 
$$ \sigma^2 \leq 1.3\cdot 2^r \,.$$
We can now finally backtrack to \eqref{eq:267} and use our high probability bound on $\sigma^2$, to get a high probability bound on $\tilde{E}$. Recall from \eqref{eq:tildeE} that 
$$\tilde{E} = \sum_{\substack{w \in \{0,1\}^r\\ j\in \{0,1\}^{n-r}}} \tilde{g}^{(j)}_w \cdot B^{wj}\,,$$ where 
the matrices $B^{wj}$, defined right after Equation~\eqref{eq:tildeE} crucially depend only on the vectors $\{\ket{g_w}\}$. Here is where the ``decoupling'' we invoked much earlier becomes useful: we can now condition on the high probability event that $\sigma^2 \leq 1.3\cdot 2^r$. Denote this event by $S$. By the independence, conditioning on $S$ does not affect the distribution of the vectors $\{\ket{\tilde{g}_w}\}$. So, we can apply the matrix Chernoff bound of Lemma~\ref{lem:matrix-chernoff} to deduce that, for any $t>0$, 
\begin{equation}
\Pr[\| \tilde{E}\| \geq t \,| \, S] \leq (2^{r+1}-1) \cdot e^{-t^2/2\sigma^2} \leq (2^{r+1}-1) \cdot e^{-t^2/(3\cdot 2^r)} \,,
\end{equation}
where the probability is now only over the vectors $\{\ket{\tilde{g}_w}\}$. We can then bound $\Pr[\| \tilde{E}\| \geq t]$ as follows:
\begin{align}
\label{eq:2677}
\Pr[\| \tilde{E}\| \geq t] &\leq \Pr[\neg S] + \Pr[\| \tilde{E}\| \geq t \,| \, S] \nonumber \\ 
&\leq 2\cdot 2^{-\frac{0.01 n^2}{4}+r} +  e^{\frac{0.01 n^2}{3}+r} +  (2^{r+1}-1) \cdot e^{-t^2/(3\cdot 2^r)}  \,. 
\end{align}
We are finally ready to prove the crucial Lemma~\ref{lem:5}. 
\begin{proof}[Proof of Lemma~\ref{lem:5}] Recall the definition of the matrix $E$ in the statement of Lemma~\ref{lem:5}. By Corollary~\ref{cor:10}, we have that, for any $t \geq 0$, 
$$\Pr\Big[ \| E\| \geq t \| \Big] \leq C \cdot \Pr\Big[ C \cdot \| \tilde{E} \|\geq t \Big] \,,$$
for some constant $C>0$. Using \eqref{eq:2677}, we can bound the RHS to get
$$\Pr\Big[ \| E\| \geq t \| \Big] \leq C \cdot \left( 2\cdot 2^{-\frac{0.01 n^2}{4}+r} +  e^{\frac{0.01 n^2}{3}+r} +  (2^{r+1}-1) \cdot e^{-t^2/(C^2 \cdot 3\cdot 2^r)} \right) \,.$$
Taking $t = \sqrt{2^{r+1.1\log n}}$ we have
$$\Pr\Big[ \| E\| \geq \sqrt{2^{r+1.1\log n}} \Big] \leq C \cdot \left( 2\cdot 2^{-\frac{0.01 n^2}{4}+r} +  e^{\frac{0.01 n^2}{3}+r} +  (2^{r+1}-1) \cdot e^{-\frac{n^{1.1}}{3 C^2}} \right) \,,$$
which implies the desired statement.
\end{proof}
With this bound in hand, we can now move to our even earlier goal of bounding $\| \tilde{M}\| $, which appears on the RHS of \eqref{eq:23}, and is defined in~\eqref{eq:m-tilde}. Since $\| \widetilde{M}\| = \| E\|^2$ as argued previously, we now have, by Lemma~\ref{lem:5}, that, with probability $1-O(2^{-n})$, 
\begin{equation}
\label{eq:mtilde-bound}
\| \widetilde{M} \| \leq 2^{r+1.1\log n} = 2^r \cdot n^{1.1} \,.
\end{equation}

Finally, substituting the latter into \eqref{eq:23}, we get that, with probability $1-O(2^{-n})$, 
\begin{equation}
\label{eq:end}
\eqref{eq:7} \leq \frac{(2+o(1))}{\sqrt{2^n}} \cdot |c^{\perp}| \sqrt{1-|c^{\perp}|^2} \cdot \sqrt{2^r \cdot n^{1.1}} = \frac{(2+o(1))}{n^{0.45}} \cdot |c^{\perp}| \sqrt{1-|c^{\perp}|^2} = O\left(\frac{1}{n^{0.45}}\right)\,.
\end{equation}
By now, we have likely forgotten where it all started! \eqref{eq:7} was one of the three terms in the expression for $\Pr[\textnormal{accept}|\mathsf{had}]$, which is what we originally set out to upper bound. We have fortunately already upper bounded the term \eqref{eq:8}, and so we are now left with bounding the last term \eqref{eq:6}. Luckily, we have already done almost all of the work to bound this term. 

Recall the definition of \eqref{eq:6}, we have 
\begin{align}
\eqref{eq:6} &= \frac{1}{2^{2r}} \sz \frac{1}{\|\ket{\hat{f}_z}\|^2} \cdot \sum_{\substack{w\neq w' \\ v\neq v'}} (-1)^{z \cdot (w\oplus w' \oplus v \oplus v')} 
\overline{c_w} c_v \overline{A_{w,w'}} A_{v,v'} \nonumber \\
&= \frac{1}{2^{2r}} \sz \frac{1}{\|\ket{\hat{f}_z}\|^2} \cdot \left| \sum_{w,w'} (-1)^{z \cdot (w\oplus w')} c_w A_{w,w'} \right|^2 \nonumber\\
&\leq \frac{(1+o(1))}{2^r} \sz \left| \sum_{w,w'} (-1)^{z \cdot (w\oplus w')} c_w A_{w,w'} \right|^2 \,,  \label{eq:end2}
\end{align}
where the inequality holds with probability $1 -O(2^{-n})$ by Lemma~\ref{lem:3}.

Notice that we have already bounded this quantity! From \eqref{eq:200}, we have
\begin{equation} \frac{1}{2^r} \sz\left| \sum_{w,w'} (-1)^{z \cdot (w\oplus w')} c_w A_{w,w'} \right|^2 \leq (1-|c^{\perp}|^2) \cdot \| M \|\,,
\label{eq:end3}
\end{equation}
where $\widetilde{M} = \| G\|_F^2 \cdot M$. Now, by Lemma~\ref{lem:gaussian-matrix}, $\| G\|_F^2 \geq (1-o(1)) 2^n$, except with doubly exponentially small probability in $n$, and $\|\widetilde{M} \|\leq 2^r \cdot n^{1.1}$ with probability $1-O(2^{-n})$ from our recent bound in \eqref{eq:mtilde-bound}. Putting these together with \eqref{eq:end2} and \eqref{eq:end3}, we get that, with probability $1-O(2^{-n})$,
\begin{equation}
    \eqref{eq:6} \leq \big(1+o(1)\big) \cdot (1-|c^{\perp}|^2) \cdot \frac{1}{n^{0.9}} = O\left(\frac{1}{n^{0.9}}\right)\,. \label{eq:end4}
\end{equation}
We are finally ready to bound $\Pr[\textnormal{accept}|\mathsf{had}]$. 
Combining Equations \eqref{eq:155}, \eqref{eq:end}, and \eqref{eq:end4}, we have that, with probability $1-O(2^{-n})$,
$$ \Pr[\textnormal{accept}|\mathsf{had}] = \eqref{eq:6} + \eqref{eq:7} + \eqref{eq:8} \leq O\left(\frac{1}{n^{0.9}}\right) + O\left(\frac{1}{n^{0.45}}\right) + (1+o(1)) \cdot |c^{\perp}|^2 \,.$$

Now, recall from \eqref{eq:score_std_2} (and the definition of $|c^{\perp}|$, right before Equation~\eqref{eq:17}) that 
\begin{equation*}
\Pr[\textnormal{accept}|\textsf{std}] = \sum_{z \in \{0,1\}^r} |c_z|^2 = 1-|c^{\perp}|^2\,.
\end{equation*}
With this, we can at last conclude the proof of Theorem~\ref{thm:2}: we have
$$ \frac12 \Pr[\textnormal{accept}|\textsf{std}] + \frac12 \Pr[\textnormal{accept}|\mathsf{had}] \leq \frac12 (1-|c^{\perp}|^2) + \frac12 (1+o(1)) \cdot |c^{\perp}|^2 = \frac12 + o(1) \,.$$

\section{\texorpdfstring{\boldmath A $1/\log(n)$-robust algorithm with single-qubit measurements}{A 1/log(n)-robust algorithm with single-qubit measurements}}

Here we replace the ideal measurement of the $O(\log n)$-qubit post-measurement state $\ket{\phi_z}$ by the recent 1-qubit-measurement test of Gupta, He, and O'Donnell.

\begin{theorem}[One-shot GHO \cite{guptaHeODonnell2025}]
    There is a single-copy algorithm making single-qubit measurements such that for all $n$-qubit mixed lab states $\rho$, the test
    \begin{itemize}
        \item Accepts with probability at least $\bra{\psi}\rho\ket{\psi}$.
        \item Rejects with probability at least $\big(1-\bra{\psi}\rho\ket{\psi}\big)/n$.
    \end{itemize}
\end{theorem}

Alone this test is not robust, but because we only need to apply it to an $O(\log n)$-qubit state, the robustness we get in our case is $O(1/\log n)$.

\begin{theorem}
    Let $\delta,\epsilon>0$ and fix any $C>1$.
    Let $\ket{\psi}$ be the target state and $\rho$ be any (mixed) lab state and consider the algorithm which repeats the following test $O(\log (n)\cdot\epsilon^{-1}\log(1/\delta))$ times (on fresh copies of $\rho$) and accept if the number of GHO rejects is below $\Theta(\log(1/\delta))$, otherwise reject.
    \begin{enumerate}
        \item With equal probability select the standard or Hadamard basis: $\mathsf{b}\sim\{\mathsf{std},\mathsf{had}\}$.
        \item Measure all but the last $C\log n$ qubits in the selected basis, obtaining outcome $z$.
        \item With $\ket{\psi^\mathsf{b}_z}$ and $\rho^\mathsf{b}_z$ denoting the respective $C\log n$-qubit marginals of the target and lab states postselected on outcome $z$ in basis $\mathsf{b}$, run the one-shot Gupta--He--O'Donnell (GHO) test on the remainder.
    \end{enumerate}
    Then for all but a $\exp(-\Omega(n))$-Haar-fraction of target states $\ket{\psi}$ this test \ldots
    \begin{itemize}
        \item Accepts with probability $1-\delta$ if $\bra{\psi}\rho\ket{\psi}\geq 1- \epsilon/\log n$,
        \item Rejects with probability $1-\delta$ if $\bra{\psi}\rho\ket{\psi}\leq 1-\epsilon$.
    \end{itemize}
    Moreover, this test comprises only single-qubit measurements, only $O(\log n)$ of which are adaptive.
\end{theorem}
\begin{proof}
    Denote by $T$ the inner test in the theorem and let $\ket{\psi}$ be a ``good'' target state in the sense we get the guarantees of Theorem \ref{thm:1}.
    Let $\mu_\rho^\mathsf{b}$ denote the distribution of outcomes from measuring $\rho$ in the basis $\mathsf{b}\in\{\mathsf{std},\mathsf{had}\}$.
    Then
    \begin{equation}
    \label{eq:logn-test-acc}
    \Pr[\text{$T$ accepts}] = \mathop{\E}_{\substack{\mathsf{b}\sim\{\mathsf{std},\mathsf{had}\}\\z\sim\mu^\mathsf{b}_\rho}}\Pr[\text{GHO accepts }\ket{\psi_z^\mathsf{b}},\rho_z^\mathsf{b}]
    \end{equation}
    which for soundness we can upper bound by
    \begin{align*}
        \eqref{eq:logn-test-acc} &\leq \mathop{\E}_{\substack{\mathsf{b}\sim\{\mathsf{std},\mathsf{had}\}\\z\sim\mu^\mathsf{b}_\rho}} 1-\frac{1-\bra{\psi_z^\mathsf{b}}\rho_z^\mathsf{b}\ket{\psi_z^\mathsf{b}}}{C\log{n}}\\
        &= 1-\frac{1-\E_{\mathsf{b},z}\bra{\psi_z^\mathsf{b}}\rho_z^\mathsf{b}\ket{\psi_z^\mathsf{b}}}{C\log n}\\
        &\leq 1-\frac{\frac12 -\frac12\bra{\psi}\rho\ket{\psi}-o(1)}{C\log n}\\
        &=1-\frac{\bra{\psi}\rho\ket{\psi}}{\big(2+o(1)\big)C\log n},
    \intertext{and for completeness,}
        \eqref{eq:logn-test-acc} &\geq \E_{\mathsf{b},z}\bra{\psi_z^\mathsf{b}}\rho_z^\mathsf{b}\ket{\psi_z^\mathsf{b}}\geq \bra{\psi}\rho\ket{\psi}\,.
    \end{align*}
    Routine Chernoff lifts these one-shot guarantees to the quoted performance of the outer algorithm.
\end{proof}

\bibliographystyle{alpha}
\bibliography{references}

\appendix
\section{Typical non-robustness of Huang--Preskill--Soleimanifar}
\label{sec:non-robust-HPS}
Here we remark that the algorithm of \cite{Huang2025} has robustness at most $1/n$ even for typical target states: for almost all targets $\ket{\psi}$, it is possible to construct hard lab states $\ket{\phi}$ which are essentially orthogonal to $\ket{\psi}$ yet are accepted with probability at least $1-1/n$.

\begin{prop}
    With probability at least $1-2^{-n}$ over a Haar-random target state $\ket{\psi}$, there exists a $\ket{\phi}$ such that
    \begin{itemize}
        \item $|\!\braket{\phi|\psi}\!|^2\leq O\!\left(\displaystyle\frac{n}{2^n}\right)$,
        \item $\ket{\phi}$ is accepted with probability at least $1-1/n$.
    \end{itemize}
\end{prop}
\begin{proof}
    Let $\ket{\psi}$ be a Haar-random $n$-qubit state and decompose it in the computational basis as
    $\ket{\psi}=\sum_x\alpha_x\ket{x}$.
    Define the state
    \[\ket{\phi}=\sum_x(-1)^{x_1}\alpha_x\ket{x}\]
    Then
    \begin{align*}
        \braket{\phi|\psi}=\sum_x|\alpha_x|^2(-1)^{x_1}=\sum_{x:x_1=0}|\alpha_x|^2-\sum_{x:x_1=1}|\alpha_x|^2=2\sum_{x:x_1=0}|\alpha_x|^2-1.
    \end{align*}
    The quantity $Z:=\sum_{x_1=0}|\alpha_x|^2$ is Beta($2^{n-1},2^{n-1}$)-distributed, so $Z$ is subgaussian with variance proxy $\Theta(2^{-n})$. By standard concentration results we get a universal constant $C>0$ such that
    \[\Pr[|\!\braket{\psi|\phi}\!|^2\leq Cn2^{-n}]=\Pr[(2Z-1)^2\leq Cn2^{-n}]\geq 1-2^{-n}\,.\]
    However, the algorithm of Huang, Preskill, and Soleimanifar only detects the difference between $\ket{\psi}$ and $\ket{\phi}$ if qubit $1$ is the last measured.
    This happens with probability $1/n$.
\end{proof}

\section{Impossibility of universal robustness for conditional comparison tests}
\label{sec:no-univ-robustness}

\begin{prop*}[Proposition~\ref{prop:no-univ-robustness} repeated]
    Let $r < n$. Consider any state certification algorithm which has the following structure:
    \begin{enumerate}[label=\textit{\roman*}.]
        \item[(i)] Measure all but $r$ qubits with single-qubit measurements $M_1,\ldots, M_{n-r}$. (Any measurement order and adaptivity is allowed.)
        \item[(ii)] Measure the remaining $r$ qubits with a measurement $M^*$ depending arbitrarily on measurements $M_1,\ldots, M_{n-r}$ and their outcomes.
        \item[(iii)] Output \textsf{accept} or \textsf{reject} depending exclusively on the outcome of $M^*$.
    \end{enumerate}
    Then, provided the algorithm accepts the target state itself with certainty, there exist target states $\ket{\psi}$ for which the algorithm has robustness $O(r/n)$.
\end{prop*}

\begin{proof}
Consider a target state that is the product of 1-qubit states:
$\ket{\psi}=\bigotimes_{j=1}^{n}\ket{\psi_j}$.
Now for each $j$ choose $\ket{\phi_j}$ with $|\braket{\phi_j|\psi_j}|^2=1-r/n$.
Then with $\ket{\phi}=\otimes_j \ket{\phi_j}$, we see
\[|\braket{\phi|\psi}|^2=(1-r/n)^{n/r}\;\longrightarrow\; 1/e\qquad\text{as}\qquad n\to\infty\,.\]
However, because $\ket{\psi}$ is a product state, the measurements $M_1,\ldots, M_{n-r}$ have no effect on the ultimate site of the $M^*$ measurement.
This means that 
\[\big|\Pr[M^* \text{ accepts on } \ket{\psi}] - \Pr[M^* \text{ accepts } \ket{\phi}]\big|\leq r/n \,.\]
Since $|\braket{\phi|\psi}|^2$ is bounded away from one by a constant (for large enough $n$), it must have constant overlap with some state $\ket{\phi'}$ orthogonal to $\ket{\psi}$ such that 
\[\big|\Pr[M^* \text{ accepts on } \ket{\psi}] - \Pr[M^* \text{ accepts } \ket{\phi'}]\big|\leq O(r/n) \,.\qedhere\]
\end{proof}
\end{document}